%% file: main.tex
\documentclass[11pt]{article}
\usepackage[margin=1in,bottom=1in]{geometry}
\usepackage[utf8]{inputenc}

\usepackage{libertine}
\usepackage{amsthm}
\usepackage{amsmath}

\usepackage{amssymb}
\usepackage{stmaryrd}
\usepackage{mathrsfs}
\usepackage{mathtools}
\usepackage{thm-restate}
\usepackage[colorlinks = true,linkcolor = blue,urlcolor = blue, citecolor = blue]{hyperref}
\usepackage[capitalize, nameinlink]{cleveref}
\usepackage{url}
\usepackage{todonotes}
\usepackage{mathrsfs}
\usepackage{soul}
\usepackage{relsize}
\usepackage[mathlines]{lineno}
\usepackage{multicol}
\usepackage{enumerate}
\usepackage[shortlabels]{enumitem}
\usepackage{xspace}
\usepackage{lineno}
\usepackage{subcaption}
\usepackage{tikz}
\usetikzlibrary{positioning}
\usetikzlibrary{arrows.meta}
\definecolor{custom-pink}{rgb}{.96,0.67,0.67}
\definecolor{custom-orange}{rgb}{.95,0.79,0.30}
\definecolor{custom-blue}{rgb}{.59,.75,.96}

\usepackage{textpos}

\usepackage{titling}

\usepackage{macros}

\begin{document}

\title{{\huge{Low rank MSO}}\thanks{%
The research of Mi.\ Pilipczuk, W. Przybyszewski, and G. Stamoulis has been supported by the project BOBR that has received funding from the European Research Council (ERC) under the European Union's Horizon 2020 research and innovation programme, grant agreement No. 948057. In particular, a majority of work on this manuscript was done while G. Stamoulis was affiliated with University of Warsaw.}}

\author{
Miko{\l}aj Boja\'nczyk\\\small{University of Warsaw}\\\small{\href{bojan@mimuw.edu.pl}{bojan@mimuw.edu.pl}}
\and
Micha{\l} Pilipczuk\\\small{University of Warsaw}\\\small{\href{michal.pilipczuk@mimuw.edu.pl}{michal.pilipczuk@mimuw.edu.pl}}
\and
Wojciech Przybyszewski\\\small{University of Warsaw}\\\small{\href{przybyszewski@mimuw.edu.pl}{przybyszewski@mimuw.edu.pl}}
\and
Marek Soko\l{}owski\\\small{Max Planck Institute for Informatics, Saarbr\"ucken}\\\small{\href{msokolow@mpi-inf.mpg.de}{msokolow@mpi-inf.mpg.de}}
\and
Giannos Stamoulis\\\small{IRIF, Université Paris Cité, CNRS, Paris}\\\small{\href{giannos.stamoulis@irif.fr}{giannos.stamoulis@irif.fr}}
}

\date{}

\maketitle 
\begin{abstract}
    We introduce a new logic for describing properties of graphs, which we call \emph{low rank \mso}. This is the fragment of monadic second-order logic in which set quantification is restricted to vertex sets of bounded cutrank. We prove the following statements about the expressive power of low rank \mso.
    \begin{itemize}[nosep]
     \item Over any class of graphs that is weakly sparse, low rank \mso has the same expressive power as separator logic. This equivalence does not hold over all graphs.
     \item Over any class of graphs that has bounded VC dimension, low rank \mso has the same expressive power as flip-connectivity logic. This equivalence does not hold over all graphs.
     \item Over all graphs, low rank \mso has the same expressive power as flip-reachability logic.
    \end{itemize}
    Here, separator logic is an extension of first-order logic by basic predicates for checking connectivity, which was proposed by Boja\'nczyk~[ArXiv 2107.13953] and by Schirrmacher, Siebertz, and Vigny [ACM ToCL 2023]. Flip-connectivity logic and flip-reachability logic are analogues of separator logic suited for non-sparse graphs, which we propose in this work. In particular, the last statement above implies that every property of undirected graphs expressible in low rank \mso can be decided in polynomial~time.
\end{abstract}

\thispagestyle{empty}

\begin{textblock}{20}(-2.1,3.8)
 \includegraphics[width=60px]{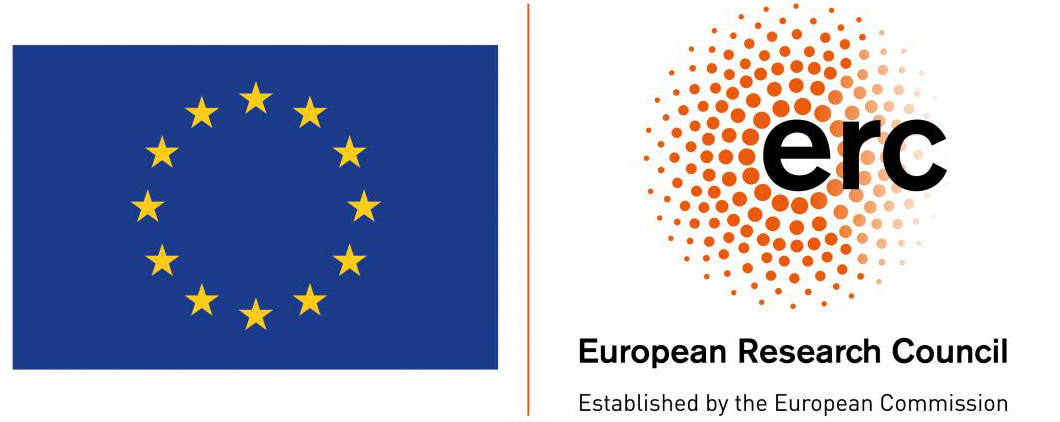}
\end{textblock}
\newpage

\pagenumbering{arabic}

\clearpage
\setcounter{page}{1}

\input{intro}
\input{logics}
\input{framework}
\input{separator}
\input{flip-connectivity}
\input{flip-reachability}

\input{conclusions}

\bibliographystyle{plain}

\end{document}

%% file: intro.tex
\section{Introduction}\label{sec:intro}

Much of the contemporary work on logic on graphs revolves around the first-order logic \fo and variants of the monadic second-order logic \mso. The general understanding is that \mso behaves well, in terms of computational aspects, on graphs that have a tree-like structure~\cite{courcelleMonadicSecondorderLogic1990,courcelle2000linear}.
The tameness of \fo is the subject of an ongoing line of work and is expected to be connected with the non-existence of complicated local structures~\cite{DreierEMMPT24,DreierMST23,DreierMT24,flippergame,Torunczyk23}. A general overview of this area can be found in the recent survey of Pilipczuk~\cite{Pilipczuk25}.

Naturally, there are interesting concepts of logic on graphs besides \fo, \mso, and their variants. One such logic was proposed independently by Boja\'nczyk~\cite{separator-logic2021} (under the name {\em{separator logic}}) and by Schirrmacher, Siebertz, and Vigny~\cite{schirrmacher2023first} (under the name \foconn). The idea is to extend \fo by a simple mechanism for expressing connectivity: we allow the usage of predicates $\conn_k(s,t,a_1,\ldots,a_k)$, for every $k\in \N$, that verify the existence of a path connecting vertices $s$ and $t$ that avoids vertices $a_1,\ldots,a_k$. Thus, the expressive power of separator logic lies strictly between \fo and \mso: the logic is not bound only to local properties, like \fo, but full set quantification is not available. Boja\'nczyk~\cite{separator-logic2021} argued that on classes of graphs of bounded pathwidth, separator logic characterizes a suitable analogue of star-free languages.
Pilipczuk, Schirrmacher, Siebertz, Toru\'nczyk, and Vigny~\cite{separatorModelChecking} showed that computational tameness of separator logic (precisely, fixed-parameter tractability of the model-checking problem) coincides on subgraph-closed classes with a natural dividing line in structural graph theory: exclusion of a fixed topological~minor.

The drawback of separator logic is that the connectivity predicates are designed with separators consisting of a bounded number of vertices in mind. Such separators are a fundamental concept in the structural theory of sparse graphs, but become ill-suited for the treatment of dense graphs. In contrast, both \fo and \mso can be naturally deployed on graphs regardless of their sparsity, and much of the recent developments around these logics concern understanding their expressive power and computational aspects on classes of well-structured dense graphs. What would be then a dense analogue of separator logic?

\vspace{-0.1cm}
\paragraph*{Low rank \mso.} We introduce a new logic on graphs called {\em{low rank \mso}} with expressive power lying strictly between \fo and \mso. We argue that low rank \mso is a natural dense analogue of separator logic. The main idea is the following: as separator logic allows quantification over small separators as understood in the theory of sparse graphs, its dense analogue should allow quantification over small separators as understood in the theory of dense graphs. The latter concept is delivered by the notion of~{\em{cutrank}}.

Let $G$ be a graph with vertex set $V$ and $X$ be a subset of $V$; denote $\wh X\coloneqq V\setminus X$ for brevity. The idea is to measure the complexity of the cut $(X,\wh X)$ by examining the adjacency matrix $\Adj_G[X,\wh X]$: the $\{0,1\}$-matrix with rows indexed by $X$ and columns indexed by $\wh X$, where the entry at the intersection of the row of $u\in X$ and the column of $v\in \wh X$ is $1$ if $u$ and $v$ are adjacent, and $0$ otherwise. The {\em{cutrank}} of~$X$, denoted $\rk(X)$, is the rank of $\Adj_G[X,\wh X]$ over the binary field $\F_2$. Cutrank was introduced by Oum and Seymour in their work on {\em{rankwidth}}~\cite{oum2005rank,OumS06}. This graph parameter, together with the functionally equivalent {\em{cliquewidth}}~\cite{courcelle2000linear} and {\em{NLC-width}}~\cite{Wanke94}, is understood as a suitable notion of tree-likeness for dense graphs. In particular, in graph theory there is an ongoing work on constructing a theory of {\em{vertex-minors}}, which is expected to be a dense counterpart of the theory of graph minors. In this theory, rankwidth is the analogue of treewidth, and cuts of low cutrank play the same role as vertex separators of bounded size. See the survey of Kim and Oum~\cite{KimO24} and the PhD thesis of McCarty~\cite{RoseThesis} for an introduction to this area.

With the motivation behind cutrank understood, the definition of low rank \mso becomes natural: it is the fragment of \mso on (undirected) graphs, where every use of the set quantifier must be accompanied by an explicit bound on the cutrank of the quantified set. The syntax uses the following~constructors:
\begin{align*}
\myunderbrace{
    \forall x \quad \exists x
}{quantification \\
over vertices}
\qquad\qquad
\myunderbrace{
    \forall X \colon r \quad \exists X \colon r
}{quantification over subsets of \\
vertices  with cutrank at most $r$}
\qquad\qquad
\myunderbrace{
    \land \quad \lor \quad \neg
}{Boolean \\ combinations}
\qquad\qquad
\myunderbrace{
    E(x,y)
}{edge relation}
\qquad\qquad
\myunderbrace{
    x \in X
}{set membership}
\end{align*}
The semantics are defined in the expected way, as a fragment of \mso with quantification over sets of vertices. (This variant of \mso is often called $\textsc{mso}_1$ in the literature, as opposed to $\textsc{mso}_2$, where quantification over edge sets is also allowed. We do not consider $\textsc{mso}_2$ here.)

The goal of this work is to robustly introduce low rank \mso, study its expressive power, and pose multiple questions about its various aspects. While we are mostly concerned about the setting of undirected graphs\footnote{Concretely, in all our results we consider the setting of vertex-colored undirected graphs, where the vertex set may be additionally equipped with a number of unary predicates.}, low rank \mso can be naturally defined on any structures where a meaningful notion of rank can be considered. We expand on this in \cref{sec:conclusions}.

\paragraph*{Our results.} It is not hard to see that predicates $\conn_k$ can be expressed in low rank \mso, hence low rank \mso is at least as expressive as separator logic (see \cref{lem:easy-comparison}). We first prove that on classes of sparse graphs, the expressive power of the two logics actually coincide. Here, we say that a class of graphs $\Cc$ is {\em{weakly sparse}} if there is $t\in \N$ such that all graphs from $\Cc$ do not contain the biclique $K_{t,t}$ as a subgraph.

\begin{theorem}
    \label{thm:main-separator}
    Let $\Cc$ be a weakly sparse class of graphs. Then for every formula of low rank \mso there exists a formula of separator logic such that the two formulas are equivalent on all graphs from $\Cc$.
\end{theorem}

\cref{thm:main-separator} corroborates the claim that low rank \mso is a dense analogue of separator logic. There is, however, an important difference. Low rank \mso is defined as a fragment of \mso, and not as an extension of \fo, like the separator logic. This has a particular impact on the computational aspects. For instance, observe that for every fixed sentence $\varphi$ of separator logic, whether $\varphi$ holds in a given a graph $G$ can be decided in polynomial time. Indeed, it suffices to apply a brute-force recursive algorithm that, for every next vertex variable quantified, examines all possible evaluations of this variable in the graph; and in the leaves of the recursion, adjacency and connectivity predicates can be checked in polynomial time. This approach does not work for low rank \mso, for the number of sets of bounded rank in a graph can be~exponential.

To bridge this gap, it would be useful to find a logic closer in spirit to \fo that would be equivalent to low rank~\mso. For this, we use the concept of {\em{definable flips}}, which has played an important role in the recent advances in understanding \fo on dense graphs, see~\cite{incremental-lemma,DreierMST23,DreierMT24,flippergame,Torunczyk23}. Formal definitions are given in \cref{sec:logics}, but in a nutshell, a {\em{definable flip}} of a graph $G$ is any graph that can be obtained from $G$ as~follows:
\begin{itemize}[nosep]
 \item Select a tuple $\tup a$ of vertices of $G$. These are the {\em{parameters}} of the flip.
 \item Classify the vertices of $G$ according to their adjacency to the vertices of $\tup a$. The vertices of $\tup a$ belong to their own singleton classes.
 \item For every pair of classes $K$ and $L$, either leave the adjacency relation between $K$ and $L$ intact, or {\em{flip}} it: exchange all edges with non-edges, and vice versa. This can be applied also for $K=L$, which amounts to complementing the adjacency relation within $K$.
\end{itemize}
Note that a single flip applied in the last point can remove a large biclique, so a definable flip of a dense graph may be much sparser. In the aforementioned works~\cite{incremental-lemma,DreierMST23,DreierMT24,flippergame,Torunczyk23}, definable flips have been identified as a useful notion of separations in dense graphs that can be described succinctly, by specifying a constant number of parameters and a constant-size piece of information on which pairs of classes should be~flipped.

This leads to the following candidate for an extension of \fo that could be equivalent to low rank~\mso. We define {\em{flip-connectivity logic}} as the extension of \fo on undirected graphs obtained by additionally allowing the usage of predicates $\flipconn_{k, A}(s, t, a_1, \ldots, a_k)$. Such a predicate verifies whether $s$ and $t$ are in the same connected component of the flip of $G$ defined by the tuple $\tup a=(a_1,\ldots,a_k)$ and the relation $A$ specifying pairs of classes to be flipped. Again, it is not hard to prove that these predicates can be expressed in low rank \mso, implying that low rank \mso is at least as expressive as flip-connectivity logic (see \cref{lem:easy-comparison}). Somewhat surprisingly, we show that the expressive power of low rank \mso is in fact strictly larger on all graphs, but the two logics are equivalent on every class of graphs with bounded Vapnik--Chervonenkis (VC) dimension (see \cref{ssec:vcdim} for the definition).

\begin{restatable}{theorem}{mainFconnNegative}\label{thm:main-fconn-negative}
    There is a sentence of low rank \mso that cannot be expressed in flip-connectivity logic.
\end{restatable}

\vspace{-0.3cm}

\begin{restatable}{theorem}{mainFconnPositive}\label{thm:main-fconn-positive}
    Let $\Cc$ be a graph class of bounded VC dimension. Then for every formula of low rank \mso there exists a formula of flip-connectivity logic such that the two formulas are equivalent on all graphs from~$\Cc$.
\end{restatable}

Finally, we show that flip-connectivity logic can be amended to a stronger logic so that it becomes equivalent to low rank \mso on all graphs. Concretely, in {\em{flip-reachability logic}}, we consider a more general notion of directed flips that work naturally in the setting of directed graphs. The definition is essentially the same as in the undirected setting, except that flipping the arc relation between classes $K$ and $L$ exchanges arcs with non-arcs in the set $K\times L$. Note that thus, directed flips work on {\em{ordered pairs}} of classes: exchanging arcs with non-arcs in $K\times L$ and in $L\times K$ are two different operations, and they can be applied or not independently. Flip-reachability logic is defined by extending \fo by predicates $\flipreach_{k, A}(s, t, a_1, \ldots, a_k)$ that test whether $t$ is {\em{reachable}} from $s$ in the directed flip defined by $\tup a=(a_1,\ldots,a_k)$ and $A$.

While flip-reachability logic naturally works in the domain of directed graphs, we can deploy it also on undirected graphs by replacing every edge with a pair of arcs directed oppositely. Note, however, that even if we start with an undirected graph $G$, a definable directed flip of $G$ is not necessarily undirected. Hence, access to the reachability relation in directed flips of $G$ may provide a larger expressive power than access to the connectivity relation in undirected flips of $G$. We prove that this is indeed the case, and the expressive power of flip-reachability logic meets that of low rank \mso.

\begin{theorem}\label{thm:main-freach}
For every formula of low rank \mso there exists a formula of flip-reachability logic such that the two formulas are equivalent on all undirected graphs.
\end{theorem}

Similarly to flip-connectivity logic, also the flip-reachability predicates are expressible in low rank \mso (see \cref{lem:easy-comparison}), hence the expressive power of the two logics is indeed equal. We also remark that in \cref{thm:main-freach} it is important that we consider low rank \mso and flip-reachability logic in the domain of undirected graphs. We currently do not know whether also on all directed graphs, the expressive power of flip-reachability logic coincides with that of (naturally defined) low rank \mso.

Similarly to separator logic, both flip-connectivity logic and flip-reachability logic are defined as extensions of \fo by predicates whose satisfaction on given arguments can be checked in polynomial time. Hence, for every fixed sentence $\varphi$ of flip-reachability logic, it can be decided in polynomial time whether a given graph $G$ satisfies $\varphi$. By combining this observation with \cref{thm:main-freach} we conclude the following.

\begin{corollary}\label{cor:xp}
 Every graph property definable in low rank \mso can be decided in polynomial time.
\end{corollary}

In the terminology of parameterized complexity, this proves that the model-checking problem for low rank \mso is {\em{slice-wise polynomial}}, that is, belongs\footnote{Formally, \cref{cor:xp} places the problem only in non-uniform $\mathsf{XP}$, because for every sentence $\varphi$ of low rank \mso we obtain a different algorithm for deciding the satisfaction of $\varphi$. However, our proof of \cref{thm:main-freach} can be turned into an effective algorithm for translating $\varphi$ into an equivalent sentence $\varphi'$ of flip-reachability logic, which can be subsequently decided by brute force. This yields a uniform $\mathsf{XP}$ algorithm for model-checking low rank \mso. We omit the details.} to the complexity class $\mathsf{XP}$.

\paragraph*{Organization.} In \cref{sec:logics} we formally introduce all the considered logics and establish some basic properties. In \cref{sec:framework} we propose a simple framework for proving equivalence with low rank \mso that is reused for all the three positive results: \cref{thm:main-separator,thm:main-fconn-positive,thm:main-freach}. Main results comparing low rank \mso with separator logic (\cref{thm:main-separator}), flip-connectivity logic (\cref{thm:main-fconn-negative,thm:main-fconn-positive}), and flip-reachability logic (\cref{thm:main-freach}) are proved in \cref{sec:separator,sec:fconn,,sec:freach}, respectively. In \cref{sec:conclusions} we discuss several open questions and possible directions for future work.

%% file: logics.tex
\section{Definitions and basic properties}\label{sec:logics}

For $k\in \N$, we denote $[k]\coloneqq \{1,\ldots,k\}$.
All graphs are finite, simple (i.e. without loops or parallel edges), and undirected, unless explicitly stated. They are also possibly {\em{vertex-colored}}: the vertex set is equipped with a number of distinguished subsets -- colors.
We treat such graphs as relational structures consisting of the vertex set, the binary symmetric edge relation (denoted $E$), and a number of unary predicates signifying colors. Note that a vertex may belong to several different colors simultaneously, or to no color at~all.

For a set $A$ of vertices of $G$, we denote by $G[A]$ the subgraph of $G$ induced by $A$, i.e. the graph with vertex set $A$, edge set consisting of all edges of $G$ with both endpoints in $A$, and colors inherited naturally.
For a vertex $v$ of $G$ and a set of vertices $A \subseteq V(G)$, we denote by $E(v, S)$ the set of neighbors of $v$ in $S$, i.e. $E(v, S) = \setof{u \in S}{uv \in E(G)}$. The {\em{neighborhood}} of a vertex is $N(v)\coloneqq E(v,V(G))$.

We assume reader's familiarity with the first-order logic \fo and monadic second-order logic \mso. For convenience, in all formulas of all the considered logics, including \fo and its extensions, we allow free set variables. If $X$ is such a set variable and $x$ is a vertex variable, then we allow membership tests of the form~$x\in X$.
By $\tup x$ we denote a tuple of vertex variables and by $\wtup X$ we denote a tuple of set variables.
For a graph $G$ and a tuple of vertices $\tup x$, we denote by $G^{\tup x}$ the set of evaluations of $\tup x$ in $G$, i.e. functions from the variables in $\tup x$ to the vertices of $G$. Similarly, for a tuple of set variables $\wtup X$, we denote by $G^{\wtup X}$ the set of evaluations of $\wtup X$ in $G$.
For brevity we might identify tuples of vertices with respective evaluations.
We say that a logic $\Ll$ is an {\em{extension}} of \fo if $\Ll$ contains \fo as its fragment.

For a formula $\varphi(x, \wtup Y, \tup z)$ of a logic $\Ll$, a graph $G$, and evaluations $\wtup B$ of $\wtup Y$ and $\tup c$ of $\tup z$, we say that $\varphi(x, \wtup B, \tup c)$ {\em{defines}} a set $A \subseteq V(G)$ in $G$ if
\[
    A = \setof{v \in V(G)}{G \models \varphi(v, \wtup B, \tup c)}.
\]
We also denote by $\varphi(G, \wtup B, \tup c)$ the set defined by $\varphi(x, \wtup B, \tup c)$ in $G$.

\paragraph*{Low rank \mso.} We have already defined low rank \mso in \cref{sec:intro}. Let us make here a few simple remarks about the choices made in the definition.

For readers not familiar with the notions of rankwidth and of cutrank, measuring the complexity of a binary matrix by its rank over $\F_2$ may not be the most intuitive choice. Let us explain that the  selection of rank over $\F_2$ is immaterial, as any similar choice would lead to a logic with the same expressive~power.

For a $\{0,1\}$-matrix $M$ and a field $\F$, by $\rkk_\F(M)$ we denote the rank of $M$ over~$\F$. Further, let the {\em{diversity}} of $M$, denoted $\dv(M)$, be the number of different rows of $M$ plus the number of different columns of~$M$. We have the following simple algebraic~fact.

\begin{lemma}\label{lem:equiv-measures}
 Let $M$ be a matrix with entries in $\{0,1\}$ and $\F$ be a finite field. Then
 \[\rkk_\F(M)\leq \rkk_\Q(M)\leq \dv(M)/2\leq |\F|^{\rkk_\F(M)}.\]
\end{lemma}
\begin{proof}
 For the first inequality, every set of columns of $M$ that is dependent over $\Q$ is also dependent over~$\F$. For the second inequality, the rank of a matrix over any field is always bounded by the number different rows, as well as by the number of different columns. For the last inequality, if the rank of $M$ over $\F$ is $k$, then the columns of $M$ are contained in the span of a base of size $k$. This span has at most $|\F|^k$ different vectors, hence $M$ has at most $|\F|^k$ different columns. A symmetric argument shows that $M$ has at most $|\F|^k$ different rows as well.
\end{proof}

\cref{lem:equiv-measures} implies that for a $\{0,1\}$-matrix $M$, whether we measure the diversity of $M$, or its rank over~$\F_2$, or its rank over any other finite field $\F$, or its rank over $\Q$, all these measurements yield values that are bounded by functions of each other. Hence, if $M$ has one of those measures bounded, then all the other measures are bounded as well. Let us also observe that testing the value of any of the considered measures can be defined in \fo.

\begin{lemma}\label{lem:inter-def}
 For every $k\in \N$ and a finite field $\F$, there is an \fo formula $\varphi_{k,\F}(X)$ that for a graph $G$ and $X\subseteq V(G)$, tests whether $\rk_\F(\Adj_G[X,\wh X])\leq k$, where $\wh X=V(G)\setminus X$. Similarly, there are formulas $\varphi_{k,\Q}(X)$ and $\varphi_{k,\dv}(X)$ that test whether $\rk_\Q(\Adj_G[X,\wh X])\leq k$ and $\dv(\Adj_G[X,\wh X])\leq k$,~respectively.
\end{lemma}
\begin{proof}
 We first construct the formula $\varphi_\F(X)$. Call two vertices $u,u'\in X$ {\em{twins}} if $u$ and $u'$ have the same neighborhood in $\wh X$; equivalently, $u$ and $u'$ define equal rows in $M\coloneqq \Adj_G[X,\wh X]$. Similarly, $v,v'\in \wh X$ are twins if they have the same neighborhood in $X$, or equivalently they define equal columns of $M$. Let $C$ be an inclusionwise maximal subset of $X$ consisting of pairwise non-twins, and similarly let $D\subseteq \wh X$ be an inclusionwise maximal subset of $\wh X$ consisting of pairwise non-twins. Note that $|C|+|D|=\dv(M)\leq 2 \cdot |\F|^{\rkk_\F(M)}$ by \cref{lem:equiv-measures}. Therefore, formula $\varphi_\F(X)$ can be constructed by (i) existentially quantifying $C$ and $D$ as sets of total size at most $|\F|^k$, (ii) checking that $C$ and $D$ are inclusionwise maximal subsets of $X$ and of $\wh X$, respectively, consisting of pairwise non-twins, and (iii) verifying that the rank over $\F$ of the minor of $M$ induced by the rows of $C$ and the columns of $D$ is at most $k$, by making a disjunction over all possible adjacency relations between the vertices of $C$ and of $D$. Formula $\varphi_{k,\Q}(X)$ can be defined in the same way (here we have $|C|+|D|\leq 2\cdot 2^k$ by \cref{lem:equiv-measures} for $\F=\F_2$), while in formula $\varphi_{k,\dv}(X)$ we only need to make sure that $|C|+|D|\leq k$.
\end{proof}

From \cref{lem:equiv-measures,lem:inter-def} we conclude that in the definition of low rank \mso, regardless whether in set quantification we require providing an explicit bound on the rank of the bipartite adjacency matrix over~$\F_2$, or over any other fixed finite field $\F$, or over $\Q$, or even on the diversity of the adjacency matrix, all these logics will have the same expressive power. This is because if we have two measures $\mu_1,\mu_2$ among the above, to quantify over $X$ with $\mu_2(\Adj_G[X,\wh X])\leq k$, it suffices to quantify over $X$ with $\mu_1(\Adj_G[X,\wh X])\leq f(k)$, where $f\colon \N\to \N$ is such that $\mu_1(M)\leq f(\mu_2(M))$ for every $\{0,1\}$-matrix $M$, and verify that indeed $\mu_2(\Adj_G[X,\wh X])\leq k$ using a formula provided by \cref{lem:inter-def}. Therefore, following the literature on rankwidth we make the arbitrary choice of defining low rank \mso using the cutrank function that relies on ranks over $\F_2$. In the remainder of this paper, we denote $\rkk\coloneqq \rkk_{\F_2}$ for brevity. Also, the cutrank of a set will be called just {\em{rank}}.

\paragraph*{Separator logic, flip-connectivity logic, and flip-reachability logic.}
Separator logic has also been introduced in \cref{sec:intro}. Recall that it is defined as the extension of \fo on graphs by predicates $\conn_k$ for $k\in \N$, each of arity $k+2$, with the following semantics: if $G$ is a graph and $s,t,a_1,\ldots,a_k$ are vertices of $G$, then $\conn_k(s,t,a_1,\ldots,a_k)$ holds in $G$ if and only if there is a path with endpoints $s$ and $t$ that does not pass through any of the vertices $a_1,\ldots,a_k$.

We now define flip-connectivity logic and flip-reachability logic. Flip-reachability logic works naturally on directed graphs ({\em{digraphs}}) and flip-connectivity logic will be a special case of the definition, hence we need to introduce some terminology on digraphs.

A~directed graph (digraph) is a~pair $G=(V, E)$ consisting of a~set $V=V(G)$ of vertices and a~set $E=E(G)$ of \emph{arcs}.
An~arc from $u$ to $v$ is denoted $\vec{uv}$ and has tail $u$ and head $v$.
We specify that digraphs do not contain self-loops, so $u \neq v$ for each arc $\vec{uv}$, and neither do they contain multiple copies of an~arc.
However, we allow parallel arcs connecting two vertices in the opposite directions.

The \emph{atomic type} of a~$k$-tuple $\tup{v} = (v_1, \ldots, v_k)$ of vertices of an~undirected graph $G$, denoted $\atp(\tup{v}) = \atp(v_1, \ldots, v_k)$, is the set of all atomic formulas of the form $x_i = x_j$, $E(x_i, x_j)$, and $U(x_i)$ for some $i,j\in \{1,\ldots,k\}$ and a unary predicate $U$ in the language of $G$, satisfied by $\tup{v}$ in $G$.
Naturally, two $p$-tuples $\tup{u}$, $\tup{v}$ of vertices satisfy the same quantifier-free formulas of first-order logic with no parameters if and only if $\atp(\tup{u}) = \atp(\tup{v})$.
We define \emph{edge type} of a~$k$-tuple $\tup v$ similarly to the atomic type, but we do not consider the unary predicates.

Let $\atp^k$ denote the set of all possible atomic types of $k$-tuples of vertices of an~undirected graph. Note that $\atp^k$ is finite and of size bounded by a function of $k$ and the number of unary predicates interpreted in $G$. Atomic types could be also naturally defined for directed graphs, but we will use them only in the undirected context.

Next, we define flips of (undirected) graphs with respect to a~tuple of parameters.
Let $k \in \N$, $A \subseteq \atp^{k+1} \times \atp^{k+1}$, $G$ be an~undirected graph and $\tup{a} = (a_1, \ldots, a_k)$ be a~$k$-tuple of vertices of $G$ -- the \emph{parameters} of the flip.
Then the \emph{$A$-flip of $G$ with parameters $\tup{a}$}, denoted $G \oplus_{\tup{a}} A$, is the directed graph with the same vertex set as~$G$, where for distinct $u, v \in V(G)$, we have
\[
    \vec{uv} \in E(G \oplus_{\tup{a}} A)\qquad\textrm{if and only if}\qquad[uv \in E(G)]\ \textrm{xor}\ [(\atp(u, \tup{a}), \atp(v, \tup{a})) \in A].
\]
Note that if the relation $A$ is symmetric, then the resulting graph is always undirected, in the sense that every arc $\vec{uv}$ is accompanied with the opposite arc $\vec{vu}$.
In this case, we will say that the $A$-flip is \emph{symmetric} and consider its result to be an undirected graph.

This definition now allows us to formally introduce the flip-reachability logic and its weaker flip-connectivity counterpart.

\begin{definition}[Flip-reachability logic]
    For $k \in \N$ and $A \subseteq \atp^{k+1} \times \atp^{k+1}$, we introduce the flip-reachability predicate
    \[
        \flipreach_{k, A}(s, t, a_1, \ldots, a_k)
    \]
    as the relation on vertices in a~graph that holds precisely when there exists a~directed path from $s$ to $t$ in $G \oplus_{\tup{a}} A$.
    Flip-reachability logic is first-order logic over graphs in which the universe is formed by the vertices of a~graph, and the available relations are the binary edge relation, the unary predicates in the language, as well as the flip-reachability predicates.
\end{definition}

\begin{definition}[Flip-connectivity logic]
    For $k \in \N$ and symmetric $A \subseteq \atp^{k+1} \times \atp^{k+1}$, we introduce the flip-connectivity predicate
    \[
        \flipconn_{k, A}(s, t, a_1, \ldots, a_k)
    \]
    as the relation on vertices in a~graph that holds precisely when $s$ and $t$ are in the same connected component of $G \oplus_{\tup{a}} A$.
    Flip-connectivity logic is first-order logic over graphs in which the universe is formed by the vertices of a~graph, and the available relations are the binary edge relation, the unary predicates in the language, as well as the flip-connectivity predicates.
\end{definition}

\paragraph*{Easy comparisons.} Let us first establish the straightforward relations between the considered logics in terms of expressive power. We first note the following.

\begin{lemma}\label{lem:freach-in-lrmso}
    Every flip-reachability predicate can be expressed in low rank \mso.
\end{lemma}
\begin{proof}
    Fix $k\in \N$, $A \subseteq \atp^{k+1} \times \atp^{k+1}$, a graph $G$, and vertices $s,t,a_1,\ldots,a_k$; denote $\tup a=(a_1,\ldots,a_k)$.
    Observe that the predicate $\flipreach_{k, A}(s, t, a_1, \ldots, a_k)$ is false if and only if there is a set $X$ of vertices of $G$ such that $s\in X$, $t\notin X$, and there are no arcs from $X$ to $\wh X$ in $G \oplus_{\tup{a}} A$, where we denote $\wh X=V(G)\setminus X$.
    Note that these conditions, for a given set $X$, can be encoded by a first-order formula taking $X,s,t,a_1,\ldots,a_k$ as free variables.
    So it remains to show that quantification over such sets $X$ can be done using a low-rank quantifier, that is, that $X$ has rank bounded by a function of~$k$.

    For this, we argue that the set $X$ of vertices of $G$ that are reachable from $s$ in $G \oplus_{\tup{a}} A$ has rank at most $|\atp^{k+1}|$.
    To see this, consider two vertices $u,u'\in X$ such that $\atp(u, \tup{a}) = \atp(u', \tup{a})$ (in $G$), as well as a vertex $v\in \wh X$. Note that $v$ is adjacent in $G$ to either both $u$ and $u'$ or to none of them, for otherwise either $\vec{uv}$ or $\vec{u'v}$ would appear as an arc in $G \oplus_{\tup{a}} A$, which would contradict the assumption that $v$ is not reachable from $s$.
    Hence, whether there is an edge in $G$ between a vertex $u$ in $X$ and a vertex $v$ in $\wh X$ depends only on the atomic type $\atp(u,\tup{a})$. It follows that the adjacency matrix $\Adj_G[X, \wh X]$ has at most $|\atp^{k+1}|$ distinct rows, hence its rank is at most $|\atp^{k+1}|$.
\end{proof}

We conclude the following.

\begin{proposition}\label{lem:easy-comparison}
 The following holds:
 \begin{itemize}[nosep]
  \item For every formula of separator logic there is an equivalent formula of flip-connectivity logic.
  \item For every formula of flip-connectivity logic there is an equivalent formula of flip-reachability logic.
  \item For every formula of flip-reachability logic there is an equivalent formula of low rank \mso.
 \end{itemize}
\end{proposition}
\begin{proof}
 For the first point, it suffices to observe that there is a~symmetric $A\subseteq \atp^{k+1} \times \atp^{k+1}$ such that for every graph $G$ and $\tup a\in V(G)^k$, the flip $G\oplus_{\tup a} A$ is equal to $G$ with all edges incident to the vertices of $\tup a$ removed; then the predicate $\conn_k(s,t,a_1,\ldots,a_k)$ is equivalent to $\flipconn_{k,A}(s,t,a_1,\ldots,a_k)$. The second point is obvious: flip-connectivity predicates are special cases of flip-reachability predicates obtained by restricting attention to symmetric relations $A$. The third point follows immediately from \cref{lem:freach-in-lrmso}.
\end{proof}

Finally, we note that flip-connectivity logic has a strictly larger expressive power than separator logic. A distinguishing property is {\em{co-connectivity}}: connectivity of the complement of the graph.

\begin{proposition}
 There is a sentence of flip-connectivity logic that verifies whether a graph is co-connected. However, there is no such sentence in separator logic.
\end{proposition}
\begin{proof}
 For the sentence of flip-connectivity logic verifying co-connectivity, we can take \[\forall s\,\forall t\, \flipconn_{0,A}(s,t),\qquad\textrm{where }A=\atp^1\times \atp^1.\]
 We are left with proving that the property cannot be expressed in separator logic.

 For the sake of contradiction, suppose there is a sentence $\varphi$ of separator logic that holds exactly in graphs whose complements are connected. Since $\varphi$ is finite, there is some number $k\in \N$ such that all connectivity predicates present in $\varphi$ are of arity at most $k+2$. Now, for every integer $n>k+2$, consider the following two graphs: $G_n$ is the complement of a cycle on $n$ vertices, and $H_n$ is the complement of the disjoint union of two cycles on $n$ vertices. Note that $G_n$ is co-connected while $H_n$ is not, hence $\varphi$ holds in $G_n$ and does not hold in $H_n$. Further, it can be easily verified that both $G_n$ and $H_n$ are $k$-connected, that is, no pair of vertices can be disconnected by a $k$-tuple of other vertices. Hence, every predicate $\conn_\ell(s,t,a_1,\ldots,a_\ell)$, for $\ell\leq k$, is equivalent on $\{G_n,H_n\colon n>k+2\}$ to the first-order formula
 \[\bigwedge_{i\in [\ell]} (s\neq a_i \wedge t\neq a_i).\]
 By replacing all the connectivity predicates with the formulas above, we turn $\varphi$ into a first-order sentence $\varphi'$ such that for all $n>k+2$, $\varphi'$ holds in $G_n$ and does not hold in $H_n$. However, a standard argument based on Ehrenfeucht--Fra\"isse games for first-order logic shows that such a sentence $\varphi'$ distinguishing $G_n$ from $H_n$ does not exist.
\end{proof}

%% file: framework.tex
\section{A general framework for proving equivalence with low rank \mso}\label{sec:framework}

In the upcoming sections we will prove that low rank \mso has the same expressive power as separator logic on weakly sparse graph classes (\cref{thm:main-separator}), as flip-connectivity logic on graph classes with bounded VC-dimension (\cref{thm:main-fconn-positive}), and as flip-reachability logic on all graphs (\cref{thm:main-freach}). All these proofs follow a similar scheme: we prove that quantification over sets of low rank can be emulated in the considered logic. The goal of this section is to provide a framework that captures common features of all three proofs, reducing the task to verifying a concrete property of the considered logic on the considered graph classes.

This key property is captured in the following definition.

\begin{definition}
    \label{def:low-rank-definability}
    We say that a logic $\Ll$ has \emph{low rank definability property} on a class of (colored) graphs $\Cc$ if for every $r \in \N$ and formula $\varphi(X, \wtup Y, \tup z)$ of $\Ll$ with free set variables $X, \wtup Y$ and free vertex variables $\tup z$, there is a formula $\psi(x, \tup t, \wtup Y, \tup z)$ of $\Ll$ with the following property: For every $G \in \Cc$, evaluation $\wtup B$ of $\wtup Y$ and evaluation $\tup c$ of $\tup z$, if there is a set $A \subseteq V(G)$ with $\rk(A) \le r$ such that $G \models \varphi(A, \wtup B, \tup c)$, then there is a set $A' \subseteq V(G)$ and a tuple $\tup d \in V(G)^{\tup t}$ such that
    \[G \models \varphi(A', \wtup B, \tup c)\qquad \textrm{and}\qquad A' = \psi(G, \tup d, \wtup B, \tup c).\]
\end{definition}

\begin{theorem}
    \label{thm:low-rank-quantifier-elimination}
    Suppose a logic $\Ll$ is an extension of \fo and has low rank definability property on a class of (colored) graphs $\Cc$. Then for every formula of low rank \mso there is a formula of $\Ll$ so that the two formulas are equivalent on all graphs from $\Cc$.
\end{theorem}
\begin{proof}
    We show that over our class of graphs $\Cc$, every formula of low rank \mso is equivalent to some formula of $\Ll$ by induction on the structure of a formula of low rank \mso.
    The only interesting part of the induction is low rank set quantification, since all other formula constructors of low rank \mso are already present in \fo.
    Consider a formula that begins with such a quantifier:
    \begin{align}
    \label{eq:existential-formula}
    \exists X \colon r\quad \varphi(X, \wtup Y, \tup z).
    \end{align}
    By induction assumption, the inner formula $\varphi$ can be expressed in $\Ll$.
    We may also assume that for every graph $G$, a set $A \subseteq V(G)$, and evaluations $\wtup B$ of $\wtup Y$ and $\tup c$ of $\tup z$, the condition $G \models \varphi(A, \wtup B, \tup c)$ implies that the rank of $A$ is at most $r$.
    Indeed, by \Cref{lem:inter-def} we know that the rank of a set $A$ in a graph $G$ is at most $r$ can be verified in \fo, and we can always add this assertion to $\varphi$.
    Then, by low rank definability property we have a formula $\psi(x, \tup t, \wtup Y, \tup z)$ of $\Ll$ such that for any graph $G \in \Cc$, whenever there exists a set $A \subseteq V(G)$ of rank at most $r$ satisfying $G \models \varphi(A, \wtup B, \tup c)$ for some evaluation $\wtup B$ of $\wtup Y$ and $\tup c$ of $\tup z$, then we have a tuple $\tup d \in V(G)^{\tup t}$ and a set $A' = \psi(G, \tup d, \wtup B, \tup c)$ such that $G \models \varphi(A', \wtup B, \tup c)$.
    Since the satisfaction of $\varphi$ implies that the rank is at most $r$, we can express the formula \eqref{eq:existential-formula} in $\Ll$.
\end{proof}

Next, we propose tools for streamlining the verification of the low rank definability property for the considered extensions of \fo. Intuitively, for each considered extension and relevant graph class $\Cc$, we will show a structure theorem showing that every set $A$ of low rank in a graph $G\in \Cc$ admits a certain structure guarded by a small set of parameters $S\subseteq V(G)$. It will be useful to consider $A$ in a suitable ``simplification'' of~$G$, which we call the {\em{$S$-operation}} of $G$. The definition of $S$-operation depends on the considered logic, as explained formally~below.

\begin{definition}
    Let $\Ll \in \set{\text{separator logic}, \text{flip-connectivity logic}, \text{flip-reachability logic}}$, $G$ be a colored graph and $S \subseteq V(G)$ be a subset of its vertices.
    We say that a graph $G'$ is an \emph{$S$-operation} with respect to $\Ll$ of $G$ if the following holds:
    \begin{itemize}[nosep]
        \item If $\Ll$ is separator logic, then $G'$ is the graph obtained from $G$ by isolating vertices from $S$ (i.e. removing all edges incident to vertices in $S$). Further, for each $s\in S$ add a unary predicate that marks the neighborhood of $S$ in $G$, and a unary predicate that marks only $s$.
        \item If $\Ll$ is flip-connectivity logic, then $G'$ is a symmetric flip of $G$ with parameters $S$. Further, for every atomic type over $S$ in $G$ (with an arbitrary fixed enumeration of $S$), add a unary predicate that marks the vertices of $G$ of this type.
        \item If $\Ll$ is flip-reachability logic, and $G'$ is a flip of $G$ with parameters $S$. Again, for every atomic type over $S$ in $G$, add a unary predicate that marks the vertices of $G$ of this type.
    \end{itemize}
\end{definition}

Observe that if we fix an arbitrary enumeration of $S$, then the unary predicates that we add to $G'$ (but not their interpretations) depend only on $|S|$.
Therefore, for a given class of graphs $\Cc$ we may talk about the language of $k$-operations, for a fixed $k \in \N$.
This is the language that consists of all the relations used in graphs from $\Cc$ and all the unary predicates added to $S$-operations of graphs from $\Cc$ for $S$ of size $k$.

Next, for each considered extension $\Ll$ of \fo, we  show that we can freely translate $\Ll$-formulas working on the graph and on its $S$-operation.

\begin{lemma}
    \label{lem:operation-forawrd}
    Let $\Ll \in \set{\text{separator logic}, \text{flip-connectivity logic}, \text{flip-reachability logic}}$, $G$ be a graph, $S \subseteq V(G)$ be a subset of its vertices, $G'$ be an $S$-operation with respect to $\Ll$ of $G$, and $\varphi(\wtup X, \tup y)$ be a formula of $\Ll$ with free set variables $\wtup X$ and free vertex variables $\tup y$. Then there is a formula $\varphi'(\wtup X, \tup y)$ of $\Ll$, depending only on $\varphi$ and $|S|$, such that for every evaluation $\wtup A$ of $\wtup X$ and $\tup b$ of $\tup y$ we have
    \[
        G \models \varphi(\wtup A, \tup b) \qquad\textrm{if and only if}\qquad G' \models \varphi'(\wtup A, \tup b).
    \]
\end{lemma}
\begin{proof}
    We start with the case when $\Ll$ is separator logic.
    Observe that in $G'$ we can define the adjacency relation of $G$ as a quantifier-free formula that uses the unary predicates for the neighborhoods of vertices in $S$ and the unary predicates for the individual vertices of $S$.
    Indeed, if we want to check whether $G \models E(u, v)$ for some vertices $u, v \in V(G)$, then first we check if any of them is in $S$.
    If not, then we know $G \models E(u, v)$ if and only if $G' \models E(u, v)$.
    In the other case, i.e.\ when (at least) one of them is in $S$, we check if the other one is marked with the respective unary predicate.

    Next, observe that we can define each connectivity predicate of $G$ in $G'$.
    Indeed, to check whether $G \models \conn_k(u, v, \tup a)$ for some vertices $u, v \in V(G)$ and vertices $\tup a \in V(G)^k$, we can guess (by making a disjunction over all cases) if the path goes through vertices from $S$ and in which order they appear.
    Then in $G'$ we can express that there are consecutive paths between $u$, guessed vertices from $S$, and $v$ after removing $\tup a$.
    In this way we get a disjunction over quantifier-free formulas that use connectivity predicates in $G'$.

    To obtain the formula $\varphi'$ we replace all occurrences of the edge relation and connectivity predicates in $\varphi$ with the respective quantifier-free formulas.

    The cases of flip-connectivity logic and flip-reachability logic are similar.
    The only interesting case is when we need to express flip-connectivity (respectively flip-reachability) predicates of $G$ in $G'$.
    For this note that since we added unary predicates to $G'$, $G$ is a symmetric flip (respectively flip) of $G'$ without parameters.
    So every symmetric flip (respectively flip) of $G$ is a symmetric flip (respectively flip) of~$G'$.
\end{proof}

\begin{lemma}
    \label{lem:operation-backward}
    Let $\Ll \in \set{\text{separator logic}, \text{flip-connectivity logic}, \text{flip-reachability logic}}$, $G$ be a graph, $S \subseteq V(G)$ be a subset of its vertices, $\tup s$ be a tuple enumerating all the elements of $S$, $G'$ be an $S$-operation of $G$ with respect to $\Ll$, and $\varphi(\wtup X, \tup y)$ be a formula of $\Ll$ with free set variables $\wtup X$ and free vertex variables $\tup y$.
    Then there is a formula $\varphi'(\wtup X, \tup y, \tup z)$ of $\Ll$, depending only on $\varphi$ and $|S|$ with $|\tup z| = |\tup s|$, such that for every evaluation $\wtup A$ of $\wtup X$ and $\tup b$ of $\tup y$ we have
    \[
        G' \models \varphi(\wtup A, \tup b) \qquad\textrm{if and only if}\qquad G \models \varphi'(\wtup A, \tup b, \tup s).
    \]
\end{lemma}
\begin{proof}
    We start with the case when $\Ll$ is separator logic.
    The proof is similar to the proof of \cref{lem:operation-forawrd}.
    We clearly see that we can define the adjacency relation of $G'$ in $G$ using the provided tuple $\tup s$.
    The same is true for the unary predicates marking the neighborhoods of vertices from $S$ in $G'$.
    Finally, for each connectivity predicate $\conn_k(u, v, \tup a)$ in $G'$ we can replace it with predicate $\conn_{k+|S|}(u, v, \tup a\tup s)$ of $G$.

    The cases of flip-connectivity logic and flip-reachability logic are similar.
\end{proof}

Intuitively,
\cref{lem:operation-forawrd,lem:operation-backward} allow us to reduce verification of the low rank definability property from the original graph to its $S$-operation, for any set of parameters $S$ of bounded size. As the $S$-operation will be a much simpler graph, arguing definability there is significantly easier. We will use this scheme in the proofs of \cref{thm:main-separator,thm:main-fconn-positive}. In the proof of \cref{thm:main-freach}, low rank definability property will be argued directly.

%% file: separator.tex
\section{Relation to separator logic}
\label{sec:separator}

In this section we prove \cref{thm:main-separator}. By \cref{thm:low-rank-quantifier-elimination}, it suffices to show that for every weakly sparse graph class $\Cc$, separator logic has low rank definability property on $\Cc$. The plan is as follows. We first prove a combinatorial characterization of low rank sets in weakly sparse graph classes. Next, we use this characterization to prove low rank definability property.

For the characterization, we prove that every low rank set in a graph from a weakly sparse graph class is very close to a (vertex) separation of small order. Let us first recall some graph-theoretic terminology. A {\em{separation}} of a graph $G$ is a pair $(L,R)$ of vertex subsets such that $L\cup R=V(G)$ and there is no edge with one endpoint in $L\setminus R$ and the other in $R\setminus L$; see \cref{fig:sep}. The {\em{order}} of the separation $(L,R)$ is $|L\cap R|$, i.e. the number of vertices the two sides have in common. We shall say that a set of vertices $X\subseteq V(G)$ is {\em{captured}} by a separation $(L,R)$ if we have
\[L\setminus R \subseteq X\subseteq L.\]
That is, $X$ has to contain all the vertices that lie strictly in the left side of the separation (set $L\setminus R$) and may additionally contain a subset of the separator (set $L\cap R$).

	\begin{figure}
\centering
		\includegraphics[page=2,scale=0.22]{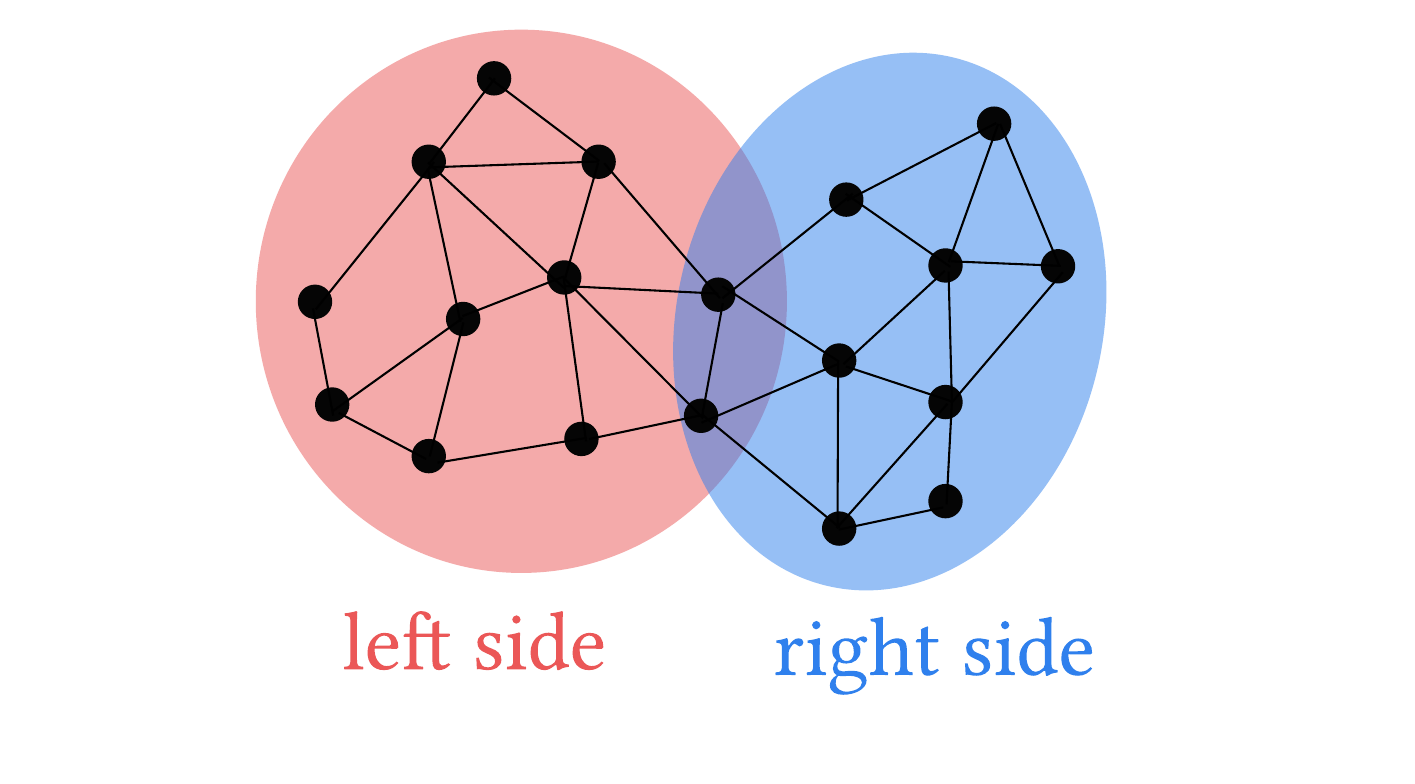}
		\caption{A separation.}\label{fig:sep}

	\end{figure}

In this terminology, the combinatorial characterization  reads as follows.

\begin{lemma}\label{lem:characterization-separator}
 Let $G$ be a graph that does not contain the complete bipartite graph $K_{t,t}$ as a subgraph. Let $X$ be a set of rank at most $r$ in $G$. Then $X$ is captured by some separation of $G$ of order at most $2^{r+1}(t-1)$.
\end{lemma}
\begin{proof}
 Denote $\wh X\coloneqq V(G)\setminus X$ and consider the bipartite adjacency matrix $M\coloneqq \Adj_G[X,\wh X]$. Call a row of $M$ {\em{frequent}} if it appears in $M$ at least $t$ times, that is, there are at least $t-1$ other rows equal to it. Define {\em{frequent columns}} of $M$ analogously. Note that the entry at the intersection of every frequent row and every frequent column must be $0$, for otherwise the at least $t$ equal rows and the at least $t$ equal columns would induce a $t\times t$ submatrix of $M$ entirely filled with $1$s, which in turn would yield a $K_{t,t}$ subgraph in $G$. This proves that the following sets $L$ and $R$ form a separation of $G$:
 \begin{itemize}[nosep]
  \item $L$ comprises all the vertices of $X$ and those vertices of $\wh X$ whose columns are non-frequent; and
  \item $R$ comprises all the vertices of $\wh X$ and those vertices of $X$ whose rows are non-frequent.
 \end{itemize}
 Clearly, $(L,R)$ captures $X$. Further, $L\cap R$ comprises all the vertices of $X$ whose rows are non-frequent and all the vertices of $\wh X$ whose columns are non-frequent. Since $\rk(X)\leq r$, the rank of $M$ over $\F_2$ is at most $r$, hence $M$ has at most $2^r$ distinct rows and at most $2^r$ distinct columns. Since every non-frequent row occurs at most $t-1$ times in $M$, $M$ has at most $2^r(t-1)$ non-frequent rows; and similarly $M$ has at most $2^r(t-1)$ non-frequent columns. We conclude that $|L\cap R|\leq 2^{r+1}(t-1)$, as required.
\end{proof}

In the next lemma we verify the low rank definability property for $r=0$ and formulas without additional free variables. This is the key step of the proof, where we use the idempotence property of separator logic in order to turn set quantification into first-order quantification. Note here that if $X$ is a set of rank $0$ in a graph $G$, then every connected component of $G$ is either entirely contained in $X$ or is disjoint with $X$. Hence, $X$ is the union of the vertex set of a collection of connected components of $G$.

\begin{lemma}\label{lem:rank0-separator}
 Let $\varphi(X)$ be a formula of separator logic with one free set variable $X$. Then there exists a formula $\varphi'(x,\tup t)$ of separator logic, where $x$ and $\tup t$ are free vertex variables, such that the following holds. For every graph $G$, if there exists a set $A\subseteq V(G)$ of rank $0$ such that $G\models \varphi(A)$, then there exists $A'\subseteq V(G)$ of rank $0$ and an evaluation $\tup d$ of variables $\tup t$ such that
 \[G \models \varphi(A')\qquad \textrm{and}\qquad A' = \varphi'(G, \tup d).\]
\end{lemma}
\begin{proof}
  Let $\Sigma$ be the set of unary predicates used by $\varphi$, and let $q$ be the quantifier rank of $\varphi$. We will use the standard notion of {\em{types}} tailored to separator logic. Concretely, for a $\Sigma$-colored graph $H$, the {\em{type}} of $H$ is the set of all sentences of separator logic in the language of $\Sigma$-colored graphs with quantifier rank at most $q+1$ that are satisfied in $H$. A standard induction argument shows that for fixed $q$ and $\Sigma$, there are only boundedly many (in terms of $q$ and $\Sigma$) non-equivalent such sentences, hence also only boundedly many types. Let $\cal T$ be the set of all types; then $|\cal T|$ is bounded by a function of $q$ and $\Sigma$. Observe that for each $\tau \in \cal T$, there is a sentence $\psi_\tau$ such that $H\models \psi_\tau$ if and only if $H$ has type $\tau$. Indeed, it suffices to take $\psi_\tau$ to be the conjunction of all the sentences belonging to $\tau$.

 Let $\cal C$ be the set of connected components of $G$ and let $\{\cal C_\tau\colon \tau\in \cal T\}$ be the partition of $\cal C$ according to the types of the individual components; that is, a component $C\in \cal C$ belongs to $\cal C_\tau$ iff the type of $C$ is $\tau$. The following claim follows from a standard argument involving Ehrenfeucht--Fra\"iss\'e games.
 
 \begin{claim}\label{cl:EF}
  There exists a constant $p$, depending only on $q$, such that the following holds. Suppose for some type $\tau\in \cal T$, $A$ contains more than $p$ components from $\cal C_\tau$, and $A$ is also disjoint with more than $p$ components from $\cal C_\tau$. Then for every $C\in \cal C_\tau$ contained in $A$, we have $G\models \varphi(A\setminus C)$.
 \end{claim}

 An immediate corollary of \Cref{cl:EF} is the following.

 \begin{claim}\label{cl:idemp}
  Suppose there exists $A\subseteq V(G)$ of rank $0$ such that $G\models \varphi(A)$.
  Then there also exists $A'\subseteq V(G)$ of rank $0$ such that the following holds:
  \begin{itemize}[nosep]
   \item $G\models \varphi(A')$; and
   \item for each $\tau \in \cal T$, $A'$ contains at most $p$ components of $\cal C_\tau$ or $A'$ contains all but at most $p$ components~of~$\cal C_\tau$.
  \end{itemize}
 \end{claim}

 It remains to construct a sentence $\varphi'(x,\tup t)$ that will define some set $A'$ whose existence is asserted by \Cref{cl:idemp}. For this, for each $\tau \in \cal T$ construct a $p$-tuple of variables $\tup t_\tau$, and let $\tup t$ be the union of $\tup t_\tau$ over all $\tau\in \cal T$. For every function $g\colon {\cal T}\to \{0,1\}$, we may use the sentences $\{\psi_\tau\colon \tau\in \cal T\}$ to write a formula $\alpha_g(x,\tup t)$ expressing the following assertion:
 supposing $x$ belongs to a component $C$ whose type is $\tau$, we have that $C$ contains a vertex of $\tup t_\tau$ if $g(\tau)=0$, and $C$ does not contain a vertex of $\tup t_\tau$ if $g(\tau)=1$. Then, we may write $\varphi'(x,\tup t)$ as the formula expressing the following:
 \begin{itemize}[nosep]
  \item for some $g\colon {\cal T}\to \{0,1\}$, we have $G\models \varphi(\alpha_g(G,\tup t))$; and
  \item $\alpha_{g^\circ}(x,\tup t)$, where $g^\circ$ is the lexicographically smallest function from $\cal T$ to $\{0,1\}$ satisfying the above.
 \end{itemize}
 By \Cref{cl:idemp}, for some evaluation $\tup d$ of $\tup t$ and some function $g\colon {\cal T}\to \{0,1\}$, we have $A'=\alpha_g(G,\tup d)$. As $G\models \varphi(A')$, the first point above holds. Then we have $\varphi'(G,\tup d)=\alpha_{g^\circ}(G,\tup d)$, where $g^{\circ}$ is as in the second point. And by construction, this set has rank $0$ and satisfies~$\varphi$.
\end{proof}

We now generalize the conclusion of \cref{lem:rank0-separator} to low rank definability property using \cref{lem:characterization-separator} in combination with \cref{lem:operation-forawrd,lem:operation-backward}.

\begin{lemma}\label{lem:lrdp-separator}
 For every weakly sparse graph class $\Cc$, separator logic has low rank definability property on $\Cc$.
\end{lemma}
\begin{proof}
 Let $r\in \N$ be a bound on the rank and $\varphi(X,\tup Y,\tup z)$ be a formula of separator logic, as in the definition of low rank definability property. We may assume that $\wtup Y=\emptyset$ and $\tup z=\emptyset$.
 Indeed, otherwise we may expand the signature by $|\wtup Y|+|\tup z|$ unary predicates marking the free variables of $\varphi$, and perform the reasoning on the modified formula $\varphi$ that treats the free variables of $\wtup Y$ and $\tup z$ as unary predicates present in the structure.
 Once we obtain a suitable formula $\varphi'$, the additional unary predicates can be interpreted back as free variables.
 Hence, from now $\varphi$ has only one free set variable $X$.

 Since $\Cc$ is weakly sparse, there is some $t\in \N$ such that no graph from $\Cc$ contains $K_{t,t}$ as a subgraph. Consider any $G\in \Cc$ and $A\subseteq V(G)$ with $\rk(A)\leq r$. By \cref{lem:characterization-separator}, there is a separation $(L,R)$ of $G$ of order at most $p\coloneqq 2^{r+1}\cdot (t-1)$ that captures $A$. Let $S\coloneqq L\cap R$; then $|S|\leq p$. Observe that since $(L,R)$ captures $A$, $A$ is the union of a subset of $S$ and a collection of connected components of $G-S$. In other words, $A$ is a set of rank $0$ in $G'$ --- the $S$-operation of $G$ with respect to the separator logic.

 By \cref{lem:operation-forawrd}, there is a formula $\psi(X)$ depending only on $\varphi$ and $p$ such that for every $B\subseteq V(G)$,
 \[G\models \varphi(B)\qquad\textrm{if and only if}\qquad G'\models \psi(B).\]
 In particular, $G'\models \psi(A)$. By \cref{lem:rank0-separator}, we may find a formula $\psi'(x,\tup t)$ depending only on $\psi$ and $p$, an evaluation $\tup d$ of variables $\tup t$, and a set $A'\subseteq V(G)$ of rank $0$ in $G'$, such that
 \[G'\models \psi(A')\qquad \textrm{and}\qquad A'=\psi'(G',\tup d).\]
 By \cref{lem:operation-backward}, there is a formula $\varphi'(x,\tup t,\tup t')$, depending only on $\psi'$ and $p$, such that for every $u\in V(G)$ and evaluation $\tup c$ of $\tup t$, we have
 \[G\models \varphi'(u,\tup c, \tup s)\qquad \textrm{if and only if}\qquad G'\models \psi'(u,\tup c),\]
 where $\tup s$ is the tuple of elements from $S$.
 All in all, we have
 \[G\models \varphi(A')\qquad \textrm{and}\qquad A'=\varphi'(G',\tup d, \tup s).\]
 Therefore, all the conditions required from $\varphi',A'$ are~met.
\end{proof}

Now \cref{thm:main-separator} follows immediately by combining \cref{lem:lrdp-separator} with \cref{thm:low-rank-quantifier-elimination}.

%% file: flip-connectivity.tex
\section{Relation to flip-connectivity logic}\label{sec:fconn}

In this section we compare the expressive power of low rank \mso with flip-connectivity logic.

\subsection{Negative result}
As our first result we show that on the class of all graphs low rank \mso is strictly more expressive than flip-connectivity logic. That is, we prove \Cref{thm:main-fconn-negative}, which we restate here for convenience.

\mainFconnNegative*
\begin{proof}
    In the proof we construct a family of graphs $G_n$ and $G_n'$ that can be distinguished by a sentence of low rank \mso, but for every sentence $\psi$ of flip-connectivity logic there is some $n$ such that $G_n$ and $G_n'$ are indistinguishable by $\psi$.

    We begin with an auxiliary graph, which we call $F_n$. In this graph, there are $n+1$ layers of vertices. There are $10n^2$ vertices in layer $0$. The vertices in layer $i \in \set{1,\ldots,n}$ are all subsets of vertices in layer $i-1$ that have size $5n$, and the edge relation between layers $i$ and $i-1$ represents membership. Thus, layer $1$ has $\binom{10n^2}{5n}$ vertices, layer $2$ has $\binom{\binom{10n^2}{5n}}{5n}$ vertices, and so on. Here is a picture:
    \mypic{4}

    Using the auxiliary graph $F_n$, we define the two colored graphs $G_n$ and $G'_n$. Actually, the underlying uncolored graphs of $G_n$ and $G'_n$ will be exactly the same; the difference between them is only in the placement of colors (unary predicates). The uncolored graph underlying $G_n$ and $G_n'$ is constructed as~follows:
        \begin{itemize}[nosep]
            \item take four disjoint copies of $H_n$;
            \item merge the first two copies of layer 0;
            \item merge the last two copies of layer 0;
            \item create a biclique between the merged copies from the previous two items.
        \end{itemize}
    Here is a picture of the graph: 
    \mypic{3}
    This completes the description of the uncolored graph underlying  $G_n$ and $G'_n$. To differentiate them, we add four unary predicates that mark vertices of $n$'th layer in each copy of $F_n$. These, call them $A, B, C, D$, are distributed as follows:
    \mypic{5}

    To simplify the notation in our proof, we denote sets of vertices marked with predicates $A, B, C,$ and $D$ by $V_n^A, V_n^B, V_n^C,$ and $V_n^D$, respectively.
    Also, the set of vertices in the $j$'th layer of the copy of $F_n$ (for $j \in [n]$) whose last layer is marked with a predicate $L \in \set{A, B, C, D}$ is denoted by $V_j^L$.
    Finally, in $G_n$ we call the set of vertices adjacent to $V_1^A$ and $V_1^B$ by $U_1$ and the set of vertices adjacent to $V_1^C$ and $V_1^D$ by $U_2$.
    Similarly, in $G_n'$ we call the set of vertices adjacent to $V_1^A$ and $V_1^C$ by $U_1$ and the set of vertices adjacent to $V_1^B$ and $V_1^D$ by $U_2$.
    This way we can assume that $G_n$ and $G_n'$ have the same vertex set (so we interpret $V_j^L$ as the same set of vertices in both $G_n$ and $G_n'$), but the edge relation is different.
    Also, we call each set $V_j^L$ or $U_1, U_2$ a \emph{block}.
    For convenience, we will say that a~block $Q$ is an~\emph{expansion} of a~block $P$ if for every subset $P' \subseteq P$ of cardinality $5n$ there exists exactly one vertex $v \in Q$ whose neighborhood in $P$ is precisely $P'$.
    We will call this vertex a~\emph{witness} of $P'$ (in $P$).

    We start with the following claim.
    \begin{claim}
        There exists a sentence $\varphi$ of low rank \mso such that $G_n \models \varphi$ and $G_n' \not \models \varphi$ for every $n \ge 3$.
    \end{claim}
    \begin{claimproof}
        Let $\varphi$ be a sentence of low rank \mso that stipulates:
        \begin{center}
            \textit{There is a set of rank $1$ that contains all the vertices marked with predicate $A$ but no vertices marked with predicate $C$.}
        \end{center}
        This sentence is clearly true in $G_n$: the sought set is $U_1 \cup \left(\bigcup_{i \in [n]} V_i^A\right) \cup \left(\bigcup_{i \in [n]} V_i^B\right)$.
        Now, we will show that $\varphi$ does not hold in $G_n'$.
        By contradiction, assume there is a set $X \subseteq V(G_n')$ of rank $1$ that contains all the vertices marked with predicate $A$ but no vertex marked with predicate $C$.
        Consider blocks $V_n^C, V_{n-1}^C, \ldots, V_1^C, U_1, V_1^A, V_2^A, \ldots, V_n^A$ ordered in this way.
        Since $V_n^C \cap X = \emptyset$ and $V_n^A \subseteq X$, there is the~earliest block $P$ in this ordering which contains at least two vertices of $X$.
        Take any two distinct vertices $u, v \in P \cap X$.
        If $P$ is $V_i^C$ for some $i < n$ then there is at most one vertex in $V_{i+1}^C \cap X$.
        However, for each subset $T$ of vertices in $V_i^C \setminus \set{u, v}$ of size $5n - 1$, there is a vertex in $V_{i+1}^C$ that is connected in $V_{i}^C$ precisely to $T \cup \set{u}$ (i.e., a~witness of $T \cup \set{u}$).
        Since we have more than $2$ such sets $T$, we get a vertex in the complement of $X$ that is connected to $u$ but not to $v$.
        Similarly, we have a vertex in the complement of $X$ that is connected to $v$ but not to $u$, so the rank of $X$ is at least $2$.
        The case of $P$ being $U_1$ is similar.

        Now consider the case when $P$ is $V_1^A$.
        Since both $u$ and $v$ have different neighborhoods in $U_1$ each of size $5n$, there are two vertices $w$ and $w'$ in $U_1$ such that $E(u, w)$, $\neg E(v, w)$, $E(v, w')$, and $\neg E(u, w')$.
        If neither $w$ nor $w'$ is in $X$ then the rank of $X$ is at least $2$.
        Therefore, assume that one of $w, w'$ (without loss of generality $w$) is in $X$.
        Then, since $w$ has more than one neighbor in $V_1^C$, it has a neighbor $z \in V_1^C\setminus X$.
        Then we have $E(w, z)$, $\neg E(v, z)$, $E(v, w')$, and $\neg E(w, w')$, so the rank of $X$ is at least $2$.
        The case of $P$ being $V_i^A$ for some $i > 1$ is analogous.
        This shows that $G_n' \not \models \varphi$.
    \end{claimproof}

    Now assume by contradiction that $\varphi$ is equivalent to a sentence $\psi$ of flip-connectivity logic.
    Assume that $\psi$ has quantifier rank less than $q$ and uses flip-connectivity predicates with at most $q$ arguments.
    We will show that $\psi$ does not distinguish $G_n$ and $G_n'$ for $n$ large enough; in particular we require $n \ge q$.

    Our first goal is to show that in $G_n$, flip-connectivity predicates can be expressed in first-order logic (and similarly in $G_n'$).
    Consider any set $S$ of vertices of $G_n$ with $|S| \leq q$ and a symmetric flip $H_n$ of $G_n$ with parameters $S$.
    We will show that all the vertices in $V(G_n) \setminus (V_n^A \cup V_n^B \cup V_n^C \cup V_n^D \cup S)$ are in the same connected component of $H_n$.
    For this we need to show a number of claims.

    \begin{claim}
        \label{cl:limited-adjacency}
        Suppose $P, Q$ are two blocks such that either $Q$ is an~expansion of $P$ or $\{P, Q\} = \{U_1, U_2\}$.
        Then every vertex $v \notin P \cup Q$ is adjacent to at most $5n$ vertices of $Q$.
        In particular, there are at most $5nq \leq 5n^2$ vertices of $Q$ that either belong to $S$ or are adjacent to $S \setminus P$.
    \end{claim}
    \begin{claimproof}
        Straightforward case analysis.
    \end{claimproof}

    \begin{claim}
        \label{cl:expansion-adjacency}
        Let a~block $Q$ be an~expansion of a~block $P$ in $G_n$.
        Then for every two vertices $u, v \in P \setminus S$ there is a~vertex $w \in Q \setminus S$ such that $w$ is adjacent to both $u$ and $v$ in $H_n$.
    \end{claim}
    \begin{claimproof}
        We will show that there are vertices $w_1, w_2, w_3, w_4 \in Q \setminus S$ that do not neighbor any vertex of $S$ in $G_n$ and exhibit all four possible neighborhoods on $\set{u, v}$.
        Therefore $w_1, w_2, w_3,$ and $w_4$ have the same atomic type on $S$ and no matter how we flip their neighborhood class with the classes of $u$ and $v$, one of them is adjacent to both $u$ and $v$ in $H_n$.

        Since all the cases are similar, we will show that there is a vertex $w_1 \in Q \setminus S$ not neighboring any vertex of $S$ in $G_n$ that is adjacent to $u$ but not to $v$.
        Observe that there are at least $\binom{|P| - 2 - q}{5n - 1} > 5n^2$ vertices of $Q$ adjacent to $u$ and non-adjacent to $\set{v} \cup (P \cap S)$: this is because for every $T \subseteq P \setminus (\{u, v\} \cup S)$ of size $5n - 1$, $Q$ contains a~witness of $T \cup \set{u}$.
        Out of these, at most $5n^2$ vertices belong to $S$ or are adjacent to $S \setminus P$ by \Cref{cl:limited-adjacency}.
        Hence there exists a~vertex $w_1 \in Q \setminus S$ that is adjacent to $u$ and non-adjacent to $\set{v} \cup (P \cap S) \cup (S \setminus P) = \set{v} \cup S$. The other cases can be argued similarly.
    \end{claimproof}

    \begin{claim}
        \label{cl:center-adjacency}
        In $H_n$ there is either an edge between $U_1 \setminus S$ and $U_2 \setminus S$ or there is an edge between $V_1^A \setminus S$ and $V_1^D \setminus S$.
    \end{claim}
    \begin{claimproof}
        Assume that in $H_n$ there is no edge between $U_1 \setminus S$ and $U_2 \setminus S$.
        Denote $T_1 = U_1 \cap S$ and $T_2 = U_2 \cap S$.
        Consider the set $P_2$ of vertices outside of $S$ neighboring (in $G_n$) all the vertices in $T_1$ and no other vertex of $S$.
        By \Cref{cl:limited-adjacency} there are at most $5n^2$ vertices in $U_2$ that neighbor a vertex in $S \setminus T_1$, so $P_2 \cap U_2$ is non-empty.
        In the same way we show that $P_2 \cap V_1^A$ is non-empty.
        Similarly, for the set $P_1$ of vertices outside of $S$ neighboring all the vertices in $T_2$ and no other vertex of $S$, both $P_1 \cap U_1$ and $P_1 \cap V_1^D$ are non-empty.
        Note that $P_1$ and $P_2$ are two (possibly equal) adjacency classes of vertices in $G_n$.
        Since there are all possible edges between $U_1$ and $U_2$ in $G_n$, but no such edges in $H_n$, we get that during the construction of $H_n$ we flipped the adjacency relation between $P_1$ and $P_2$.
        Hence in $H_n$ there is an edge between $V_1^A \setminus S$ and~$V_1^D \setminus S$.
    \end{claimproof}

    Combining all the claims above we get the following.

    \begin{claim}
        \label{cl:one-large-component}
        All vertices of $V(H_n) \setminus \left(S \cup V_n^A \cup V_n^B \cup V_n^C \cup V_n^D\right)$ are in the same connected component of~$H_n$.
    \end{claim}
    \begin{claimproof}
        For every $i \in [n - 1]$ and $L \in \set{A, B, C, D}$, the vertices of $V_i^L \setminus S$ are in the same connected component of $H_n$: by \Cref{cl:expansion-adjacency}, each pair of vertices in $V_i^L \setminus S$ shares a~neighbor in $V_{i+1}^L \setminus S$.
        The same conclusion is true also for $U_1 \setminus S$ and $U_2 \setminus S$.
        By the same claim, there exists at least one edge between $V_i^L \setminus S$ and $V_{i+1}^L \setminus S$; at least one edge between $U_1 \setminus S$ and each of $V_1^A \setminus S$ and $V_1^B \setminus S$; and at least one edge between $U_2 \setminus S$ and each of $V_1^C \setminus S$ and $V_1^D \setminus S$.
        Finally, by \Cref{cl:center-adjacency} there is also either an edge between $U_1 \setminus S$ and $U_2 \setminus S$ or an edge between $V_1^A \setminus S$ and $V_1^D \setminus S$.
    \end{claimproof}

    Note that all the claims above also work in the setting where we consider a~symmetric flip $H'_n$ of $G'_n$ with parameters $S$. This is because the uncolored graph underlying $G_n$ and $G_n'$ is the same.

    \begin{claim}
        \label{cl:flipconn-to-fo}
        For every predicate $\flipconn_{q, \Pi}(s, t, a_1, \ldots, a_q)$ there is an \fo formula $\xi_{q, \Pi}(s, t, a_1, \ldots, a_q)$ such that for any $n$ large enough, for any vertices $s, t, a_1, \ldots, a_q$ in $G_n$ we have
        \[
            G_n \models \flipconn_{q, \Pi}(s, t, a_1, \ldots, a_q)\qquad \textrm{if and only if}\qquad G_n \models \xi_{q, \Pi}(s, t, a_1, \ldots, a_q);
        \]
        and also for any vertices $s', t', a_1', \ldots, a_q'$ in $G_n'$, we have
        \[
            G_n' \models \flipconn_{q, \Pi}(s', t', a_1', \ldots, a_q')\qquad \textrm{if and only if}\qquad  G_n' \models \xi_{q, \Pi}(s', t', a_1', \ldots, a_q').
        \]
    \end{claim}
    \begin{claimproof}
        Let $S = \{a_1, \ldots, a_q\}$ and $H_n$ be the $\Pi$-flip of $G_n$ with parameters $a_1, \ldots, a_q$.
        As in \Cref{cl:one-large-component}, all vertices of $X \coloneqq V(H_n) \setminus \left(S \cup V_n^A \cup V_n^B \cup V_n^C \cup V_n^D\right)$ belong to a~single connected component of $H_n$.
        Next, vertices of each $V_n^L \setminus S$ for $L \in \set{A, B, C, D}$ can be partitioned into at most $2^q$ parts according to the atomic types on $\set{a_1, \ldots, a_q}$; hence $(V_n^A \cup V_n^B \cup V_n^C \cup V_n^D) \setminus S$ is partitioned into $4 \cdot 2^q$ parts in total.
        All these parts are homogeneous (i.e. fully adjacent or fully non-adjacent) with respect to the adjacency in~$H_n$.
        It follows that if there exists a~path in $H_n$ between two vertices of $(V_n^A \cup V_n^B \cup V_n^C \cup V_n^D) \setminus S$ whose all vertices are in $(V_n^A \cup V_n^B \cup V_n^C \cup V_n^D) \setminus S$, then the shortest such path has length at most $4 \cdot 2^q$.

        Note now that the diameter of every connected component $Y$ of $H_n$ that is disjoint from $X$ is bounded from above by $(q + 1)(4 \cdot 2^q + 1)$. Indeed, suppose $u, v \in Y$ are connected in $H_n$ and consider a~shortest path $P$ between $u$ and $v$.
        Then every subpath of $P$ of length $4 \cdot 2^q + 1$ must contain a~vertex outside of $V_n^A \cup V_n^B \cup V_n^C \cup V_n^D \cup X$; but such a~vertex must belong to $S$ and so there are at most $q$ such vertices.
        
        Now, we can easily see that $G_n \not\models \flipconn_{q, \Pi}(s, t, a_1, \ldots, a_q)$ if and only if $s$ and $t$ reside in different connected components of $H_n$, at least one of which has diameter at most $(q + 1)(4 \cdot 2^q + 1)$; this property can be easily tested by an \fo formula, whose negation can be taken as $\xi_{q, \Pi}$.
        Also, it can be verified that the same formula $\neg \xi_{q, \Pi}$ checks if $s'$ and $t'$ reside in different connected components of $H'_n$, where $H'_n$ is the $\Pi$-flip of $G'_n$ with parameters $a_1', \ldots, a_q'$.
        Applying the same argumentation as in $H_n$, we conclude that $G_n' \models \flipconn_{q, \Pi}(s', t', a_1', \ldots, a_q')$ if and only if $G_n' \models \xi_{q, \Pi}(s', t', a_1', \ldots, a_q')$.
    \end{claimproof}

    To finish the argument, observe that if $G_n$ and $G_n'$ were distinguishable by a sentence of flip-connectivity logic, then by \Cref{cl:flipconn-to-fo} they would be distinguishable by a sentence of first-order logic.
    However, by a standard argument using Ehrenfeucht--Fraïssé games we know that $G_n$ and $G_n'$ are indistinguishable by \fo sentences of quantifier rank $o(\log n)$.
    This finishes the example and shows that on all graphs, low rank \mso is strictly more expressive than flip-connectivity logic.
\end{proof}

\subsection{VC dimension and related notions} \label{ssec:vcdim}
Before we continue to our positive result for flip-connectivity logic, i.e. \cref{thm:main-fconn-positive} that states that for every class of bounded Vapnik--Chervonenkis (VC) dimension, low rank \mso and flip-connectivity logic are equivalent, we start with a number of definitions that explain this notion.

A {\em{set system}} over a universe $U$ is just a family $\Ff$ of subsets of $U$. For a subset of the universe $X\subseteq U$, we say that $X$ is {\em{shattered}} by $\Ff$ if for every $Y\subseteq X$ there exists $F\in \Ff$ such that $F\cap X=Y$. The {\em{VC dimension}} of $\Ff$ is the largest cardinality of a set shattered by $\Ff$. The VC dimension of a graph $G$ is the VC dimension of the set system of neighborhoods $\{N(v)\colon v\in V(G)\}$; and a graph class $\Cc$ has {\em{bounded VC dimension}} if there is $d\in \N$ such that every member of $\Cc$ has VC dimension at most $d$.

Next, we will need the following definition of duality of a binary relation, and its connection to the notion of VC dimension. The following definitions and results are taken from~\cite{incremental-lemma}.

\begin{definition}
    Let $E \subseteq A \times B$ be a binary relation.
    We say that $E$ has a {\em{duality}} of order $k$ if at least one of two cases holds:
    \begin{itemize}[nosep]
        \item there exists $A_0 \subseteq A$ of size at most $k$ such that for every $b \in B$ there is some $a \in A_0$ with $\neg E(a, b)$, or
        \item there exists $B_0 \subseteq B$ of size at most $k$ such that for every $a \in A$ there is some $b \in B_0$ with $E(a, b)$.
    \end{itemize}
\end{definition}

\begin{theorem}[see {\cite{incremental-lemma}}]
    \label{thm:vc-dim-duality}
    For every $d \in \N$ there is some $k \in \N$ such that the following holds.
    Let $E \subseteq A \times B$ be a binary relation with both $A$ and $B$ finite such that the set system $\{\setof{b \in B}{E(a, b)}\colon a \in A\}$ has VC dimension at most $d$.
    Then $E$ has a duality of order $k$.
\end{theorem}

In the proof of equivalence of low rank \mso and flip-connectivity logic, we will use the following result proven by Bonnet, Dreier, Gajarsk\'y, Kreutzer, M\"ahlmann, Simon, and Toru\'nczyk~\cite{incremental-lemma}.
Here, a {\em{pseudometric}} is a symmetric function $\delta\colon V \times V \to \R_{\geq 0} \cup \set{+\infty}$ satisfying the triangle inequality; and for a partition $\Pp$ of a set $V$, by $\Pp(v)$ we denote the unique part of $\Pp$ containing $v$.

\begin{theorem}
    [{\cite[Theorem~3.5]{incremental-lemma}}]
    \label{thm:incremental-lemma}
    Fix $r, k, t \in \N$.
    Let $V$ be a finite set equipped with:
    \begin{itemize}[nosep]
        \item a binary relation $E \subseteq V \times V$ such that for all $A \subseteq V$ and $B \subseteq V$, $E \cap (A \times B)$ has a duality of order~$k$,
        \item a pseudometric $\dist\colon V \times V \to \R_{\geq 0} \cup \set{+\infty}$, and
        \item a partition $\Pp$ of $V$ with $|\Pp| \leq t$.
    \end{itemize}
    Suppose further that $E(u, v)$ depends only on $\Pp(u)$ and $\Pp(v)$, for all $u, v$ with $\dist(u, v) > r$. (That is, whenever $\dist(u,v)>r$, $\dist(u',v')>r$, $\Pp(u)=\Pp(u')$, and $\Pp(v)=\Pp(v')$, we have $E(u,v)\iff E(u',v')$.)
    Then there is a set $S \subseteq V$ of size $\Oh(kt^2)$ such that $E(u, v)$ depends only on $E(u, S)$ and $E(S, v)$ for all $u, v \in V$ with $\dist(u, v) > 5r$.
\end{theorem}

\subsection{Positive result}

Now we proceed to the proof of \Cref{thm:main-fconn-positive}.
We restate it here for convenience.
\mainFconnPositive*

We start with the following combinatorial lemma that is the analogue of \cref{lem:characterization-separator}. Intuitively, it characterizes any set of low rank in a graph $G$ of bounded VC dimension as the union of a collection of connected components in a graph obtained from $G$ by a small flip.
\begin{lemma}
    \label{lem:low-rank-sets-in-flipconn}
    Fix a constant $r \in \N$ and a class  $\Cc$ of graphs of bounded VC dimension.
    Then there exists a constant $\ell \in \N$ such that for every graph $G \in \Cc$ and every set $A \subseteq V(G)$ of rank at most $r$, there is a set $S \subseteq V(G)$ of size at most $\ell$ and a symmetric flip $G'$ of $G$ with parameters $S$ such that $A$ is the union of the vertex sets of a collection of connected components of $G'$.
\end{lemma}

\begin{proof}
    Since $\Cc$ has bounded VC dimension, by \cref{thm:vc-dim-duality} there is some $k \in \N$ such that the edge relation of every graph $G \in \Cc$ has a duality of order $k$.
    Take a graph $G \in \Cc$ and a subset $A \subseteq V(G)$ of rank at most $r$.
    Partition $A$ according to the edge types on the complement of $A$, i.e. two vertices $u, v \in A$ are in the same part if for every $w \in V(G) \setminus A$ we have $E(u, w) \iff E(v, w)$.
    Since $A$ has rank at most $r$, the number of parts is at most $2^r$.
    Similarly, partition the complement of $A$ according to the edge types on $A$; again we get at most $2^r$ parts.
    In this way we obtain a partition $\Pp$ of $V(G)$ into at most $2^{r + 1}$ parts.

    Define a pseudometric on $V(G)$ as follows:
    \[
        \dist(u, v) =
        \begin{cases}
            0 & \text{if $u, v \in A$,} \\
            0 & \text{if $u, v \in V(G) \setminus A$,} \\
            +\infty & \text{otherwise.}
        \end{cases}
    \]

    Note that $E(u, v)$ depends only on $\Pp(u)$ and $\Pp(v)$, for all $u, v$ with $\dist(u, v) > 1$.
    Therefore, by \Cref{thm:incremental-lemma} there is a set $S \subseteq V(G)$ of size at most $\ell$ for some constant $\ell = \ell(k, r)$ such that $E(u, v)$ depends only on $E(u, S)$ and $E(v, S)$ for all $u, v \in V(G)$ with $\dist(u, v) > 5$.
    In particular, for every $u \in A$ and $v \in V(G) \setminus A$ we have that $E(u, v)$ depends only on $E(u, S)$ and $E(v, S)$.
    Therefore, there exists a symmetric flip $G'$ of $G$ with parameters $S$ that has no edges between $A$ and $V(G) \setminus A$.
    Naturally, then $A$ is the union of the vertex sets of a collection of connected components of $G'$.
\end{proof}

With the combinatorial characterization of low rank sets in graphs of bounded VC dimension given by \Cref{lem:low-rank-sets-in-flipconn}, we can prove that flip-connectivity logic has low rank definability property over any class of bounded VC dimension, as postulated in \Cref{def:low-rank-definability}.
Similarly as in \cref{sec:separator}, we first prove this property just for rank $0$ sets, i.e.\ unions of the vertex sets of connected components.

\begin{lemma}
    \label{lem:rank-0-definable}
    Let $\varphi(X, \wtup Y, \tup z)$ be a formula of flip-connectivity logic with free set variables $X, \wtup Y$ and free vertex variables $\tup z$.
    Then there is a formula $\psi(x, \tup t, \wtup Y, \tup z)$ of flip-connectivity logic such that the following holds:
    For every graph $G$ and evaluations $\wtup B$ of $\wtup Y$ and $\tup c$ of $\tup z$, if there is a set $A \subseteq V(G)$ of rank $0$ such that $G \models \varphi(A, \wtup B, \tup c)$, then there exists a set $A' \subseteq V(G)$ of rank $0$ and an evaluation $\tup d$ of variables $\tup t$ such that
    \[G \models \varphi(A', \wtup B, \tup c)\qquad \textrm{and}\qquad A' = \psi(G, \tup d, \wtup B, \tup c).\]
\end{lemma}
\begin{proof}
    The proof of this lemma is essentially the same as that of \Cref{lem:rank0-separator}.
    The only substantial difference is that we need to tailor the standard notion of types to flip-connectivity logic.
    Concretely, instead of considering the set of all sentences of separator logic over a given signature and of bounded quantifier rank, we consider the set of all flip-connectivity sentences  over that signature and with the same bound on their quantifier rank.
\end{proof}
    
Now we can prove that flip-connectivity logic has low rank definability property on every class of graphs of bounded VC dimension. The proof essentially repeats the reasoning of \cref{lem:lrdp-separator}.
\begin{lemma}
    \label{lem:property-for-fconn}
    For every graph class $\Cc$ of bounded VC dimension, flip-connectivity logic has low rank definability property on $\Cc$.
\end{lemma}
\begin{proof}
    Let $r \in \N$ be a bound on the rank and $\varphi(X, \wtup Y, \tup z)$ be a formula of flip-connectivity logic, as in the definition of low rank definability property.
    By \Cref{lem:operation-forawrd} there is a formula $\varphi'(X, \wtup Y, \tup z)$ depending only on $\varphi$ and $r$ such that for every subset $S$ of vertices of $G$ of size at most $\ell$ and every symmetric flip $G'$ of $G$ with parameters $S$, we have
    \[
        G \models \varphi(A, \wtup B, \tup c)\qquad\textrm{if and only if}\qquad G' \models \varphi'(A, \wtup B, \tup c).
    \]
    Fix an evaluation $\wtup B$ of $\wtup Y$ and $\tup c$ of $\tup z$.
    Assume that there is a set $A \subseteq V(G)$ of rank at most $r$ such that $G \models \varphi(A, \wtup B, \tup c)$.
    By \Cref{lem:low-rank-sets-in-flipconn} there is a set $S \subseteq V(G)$ of size at most $\ell$ and an $S$-flip $G'$ of $G$ such that $A$ is a set of rank $0$ in $G'$.
    Also, we know that $G' \models \varphi'(A, \wtup B, \tup c)$.
    So, by \Cref{lem:rank-0-definable} there is a set $A'$ of rank $0$ in $G'$ such that $G' \models \varphi'(A', \wtup B, \tup c)$ and $A' = \psi'(G', \tup d, \wtup B, \tup c)$ for some tuple of parameters $\tup d$ and a formula $\psi'$ that depends only on $\varphi'$ (so only on $\varphi$ and $\ell$).
    By \Cref{lem:operation-backward}, there is a formula $\psi(x, \tup t, \tup t', \wtup Y, \tup z)$ such that $A' = \psi(G, \tup d, \tup s, \wtup B, \tup c)$.
    As all the conditions required from $A', \psi$ are~met,  flip-connectivity logic has low rank definability property on $\Cc$.
\end{proof}

\Cref{thm:main-fconn-positive} follows now from \Cref{lem:property-for-fconn} and \Cref{thm:low-rank-quantifier-elimination}.

%% file: flip-reachability.tex
\section{Relation to flip-reachability logic}\label{sec:freach}

In this section we prove \Cref{thm:main-freach}. The main part of the work is presented in \cref{sec:low-rank-comb}, where we develop a combinatorial characterization of sets of low rank in undirected graphs. This description is captured in the key technical statement: Low Rank Structure Theorem (\cref{thm:low-rank-structure}). Then, we complete the proof of \Cref{thm:main-freach} in \cref{sec:flip-reach-proof}.

\input{low-rank-combinatorics}

\input{low-rank-and-flip-reach}

%% file: low-rank-combinatorics.tex
\subsection{Combinatorics of sets of low rank}
\label{sec:low-rank-comb}

For an undirected graph $G$ and $r \in \N$, we define $\LowRank^r(G)$ to be the family of all vertex sets in $G$ of rank at most $r$:
\[
    \LowRank^r(G) \coloneqq \{ X \subseteq V(G) \,\mid\, \rk(X) \leq r \}.
\]
Our main goal in this section is to give an~effective representation of $\LowRank^r(G)$.
More precisely, we will show that $\LowRank^r(G)$ can be written as the~union of $|V(G)|^{\Oh_r(1)}$ \emph{well-structured} families of subsets of $V(G)$, and that each such family can be defined using $\Oh_r(1)$ vertices of $G$.
Formal definitions follow.

Let $G$ be an~undirected graph.
A~\emph{seed} of $G$ is a~tuple $(X_+, X_-, \Xc)$, where
\begin{itemize}[nosep]
 \item $X_+, X_- \subseteq V(G)$,
 \item $\Xc$ is a~collection of nonempty subsets of $V(G)$ (called {\em{parts}}), and
 \item $\Xc \cup \{X_+, X_-\}$ is a~partition of $V(G)$.
\end{itemize}
The family \emph{spanned} by a~seed $(X_+, X_-, \Xc)$ is
\[ \Span(X_+, X_-, \Xc) \coloneqq \{X_+ \cup \, \bigcup \mathcal{Y} \,\colon\, \mathcal{Y} \subseteq \Xc\}. \]
In our characterization of  $\mathsf{LowRank}^r(G)$, the obtained families covering  $\mathsf{LowRank}^r(G)$ can be spanned by seeds with two additional structural restrictions: first, they need to satisfy a~specific uniformity condition, which we will introduce in a~moment; and second, the seeds themselves need to be defined in terms of at most $\Oh_r(1)$ vertices of $G$ via fixed formulas of flip-reachability logic.

For any seed $(X_+, X_-, \Xc)$ and $u \in \bigcup \Xc$, we define $\Xc(u) \in \Xc$ as the unique part of $\Xc$ containing $u$.
Then, given a~tuple $\tup{a}$ of vertices of $G$, we say that a~seed $(X_+, X_-, \Xc)$ is \emph{$\tup{a}$-uniform} if for every $u_1, u_2 \in \bigcup \Xc$ with $\atp(u_1, \tup{a}) = \atp(u_2, \tup{a})$, we have that $N_G(u_1) \setminus (\Xc(u_1) \cup \Xc(u_2)) = N_G(u_2) \setminus (\Xc(u_1) \cup \Xc(u_2))$.
In other words, whenever $u_1$ and $u_2$ are twins with respect to $\tup{a}$, they are also twins with respect to $V(G) \setminus (\Xc(u_1) \cup \Xc(u_2))$.

Next, we say that a~seed $(X_+, X_-, \Xc)$ is \emph{defined} by formulas $\varphi_+$, $\varphi_-$, $\varphi_{\sim}$ of flip-reachability logic and a~$k$-tuple of vertices $\tup{a}$ if:
\begin{itemize}[nosep]
    \item $X_+ = \{v \in V(G) \,\colon\, G \models \varphi_+(v, \tup{a})\}$,
    \item $X_- = \{v \in V(G) \,\colon\, G \models \varphi_-(v, \tup{a})\}$, and
    \item two vertices $u, v \in V(G) \setminus (X_+ \cup X_-)$ belong to the same part of $\Xc$ if and only if $G \models \varphi_\sim(u, v, \tup{a})$.
\end{itemize}
Note that, assuming that $\varphi_+(\cdot, \tup{a})$ and $\varphi_-(\cdot, \tup{a})$ describe disjoint subsets of $V(G)$ and $\varphi_{\sim}(\cdot, \cdot, \tup{a})$ is an~equivalence relation on $V(G)$, there exists a~unique seed defined by $\varphi_+, \varphi_-, \varphi_{\sim}$ and $\tup{a}$.
We denote this seed by $\Seed(\varphi_+, \varphi_-, \varphi_{\sim}, \tup{a})$.

We are now ready to state our structural characterization of vertex sets of bounded rank.
\begin{theorem}[Low Rank Structure Theorem]
    \label{thm:low-rank-structure}
    Let $r \in \N$.
    Then there exist an~integer $k \coloneqq k(r)$ and formulas $\varphi_+$, $\varphi_-$, $\varphi_{\sim}$ of flip-reachability logic with the following properties:
    \begin{itemize}[nosep]
        \item For any undirected graph $G$ and $k$-tuple $\tup{a} \in V(G)^k$, we have that $\varphi_+, \varphi_-, \varphi_{\sim}, \tup{a}$ define an~$\tup{a}$-uniform seed of $G$.
        \item For any undirected graph $G$,
        \[
            \LowRank^r(G) = \bigcup_{\tup{a} \in V(G)^k} \Span(\Seed(\varphi_+, \varphi_-, \varphi_{\sim}, \tup{a})).
        \]
    \end{itemize}
\end{theorem}

The remainder of this section is devoted to the proof of the Low Rank Structure Theorem (\cref{thm:low-rank-structure}).

\subsubsection{Sets of low rank described as suffixes of flips}

As the first step towards the proof of \cref{thm:low-rank-structure}, we show that the family of low-rank vertex sets in $G$ can be characterized as \emph{suffixes} of certain digraphs that are flips of $G$ parameterized by a~bounded number of vertices.
Namely, our goal is to prove the following statement:

\begin{lemma}
    \label{lem:low-rank-to-suffixes}
    Let $r \in \N$.
    Then there exist an~integer $k \coloneqq k(r)$ and a~binary relation $A \subseteq \atp^{k+1} \times \atp^{k+1}$ such that, for any undirected graph $G$, we have
    \[
        \LowRank^r(G) = \bigcup_{\tup{a} \in V(G)^k} \Suffixes(G \oplus_{\tup{a}} A),
    \]
    where $\Suffixes(H)$ denotes the set of suffixes of a~directed graph $H$: \[\Suffixes(H) \coloneqq \{S \subseteq V(H) \,\colon\, \forall_{\vec{uv} \in E(H)}\, u \in S \Rightarrow v \in S\}.\]
\end{lemma}

The proof of \Cref{lem:low-rank-to-suffixes} will roughly proceed as follows.
Suppose $G$ is an~undirected graph.
Given a~set $S \subseteq V(G)$, we say that its subset $R_S \subseteq S$ is a~\emph{representative} of $S$ if every vertex $v \in S \setminus R_S$ has a~twin $r_v \in R_S$ with respect to $V(G) \setminus S$.
Equivalently,
\[
    \{ N(v) \setminus S \,\colon\, v \in R_S \} \,=\, \{ N(v) \setminus S \,\colon\, v \in S \}.
\]
Let now $X \subseteq V(G)$ be a set of vertices of rank at most $r$.
We invoke the following standard fact:
\begin{proposition}[\cite{OumS06}]
    \label{prop:small-representative}
    There exists a~representative $R_+$ of $X$ of cardinality at most $2^r$, as well as a~representative $R_-$ of $V(G) \setminus X$ of cardinality at most $2^r$.
\end{proposition}
So let us choose the representatives $R_+, R_-$ as above, and assume that these are inclusionwise minimal.
Our aim is now to recover $X$ given $R_+$ and $R_-$; or more precisely, find all possible sets $X \subseteq V(G)$ with $R_+$ forming a~representative of $X$ and $R_-$ forming a~representative of $V(G) \setminus X$.
To this aim, observe that every vertex $v \in V(G)$ has at most one twin $r_v^+ \in R_+$ with respect to $V(G) \setminus X$ (and such a~twin definitely exists when $v \in X$), and similarly, $v$ has at most one twin $r_v^- \in R_-$ with respect to $X$ (and it exists whenever $v \notin X$).
We can now easily reason about some of the vertices in $G$: for instance, a~vertex $v$ must belong to $X$ if $v \in R_+$ or $v$ has no twin in $R_-$ with respect to $X$; and $v$ must not belong to $X$ if $v \in R_-$ or $v$ has no twin in $R_+$ with respect to $V(G) \setminus X$.
What happens with the remaining vertices is a~bit trickier: we will show that there exists a~flip $H \coloneqq G \oplus_{\tup{a}} A$, with $\tup{a}$ consisting precisely of the vertices of $R_+ \cup R_-$, with the property that whenever $\vec{uv}$ is an~arc of $H$, then $u \in X$ implies $v \in X$.
The precise details follow below.

\medskip

We begin by explaining the construction of the flip $G \oplus_{\tup{a}} A$.
Fix $r \in \N$ for the course of the proof and set $k \coloneqq 2 \cdot 2^r$.
We consider every $k$-tuple $\tup{a} \in V(G)^k$ of parameters to consist of two halves $\tup{a}^+ = (a^+_1, \ldots, a^+_{2^r})$, $\tup{a}^- = (a^-_1, \ldots, a^-_{2^r})$; intuitively, for a~low-rank subset $X$ of $V(G)$, $\tup{a}^+$ denotes a~representative set of $X$ with respect to $V(G) \setminus X$, and $\tup{a}^-$ denotes a~representative set of $V(G) \setminus X$ with respect to $X$.
In the following description, we will often slightly abuse the notation and use $\tup{a}^+$ and $\tup{a}^-$ to denote the set of vertices present in the respective tuple of vertices.

For every $v \in V(G)$, we can deduce from the atomic type $\atp(v, \tup{a}^+, \tup{a}^-) \in \atp^{k+1}$ in $G$ the following information:
\begin{itemize}[nosep]
    \item whether $\tup{a}^+$ and $\tup{a}^-$ are disjoint;
    \item the subgraph $G[\tup{a}]$ of $G$ induced by $\tup{a}$;
    \item whether in $G[\tup{a}]$, the set $\tup{a}^+$ of vertices has rank at most $k$;
    \item whether $v \in \tup{a}^+$ and whether $v \in \tup{a}^-$;
    \item a~vertex $\varphi^+(v) \in \tup{a}^+$ with the property that $N(v) \cap \tup{a}^- = N(\varphi^+(v)) \cap \tup{a}^-$; we place $\varphi^+(v) = \bot$ if no such vertex exists, and in case multiple such vertices exist, we choose $\varphi^+(v)$ as the earliest entry of $\tup{a}^+$ with this property;
    \item an~analogous vertex $\varphi^-(v) \in \tup{a}^-$ with the property that $N(v) \cap \tup{a}^+ = N(\varphi^-(v)) \cap \tup{a}^+$ (again, we place $\varphi^-(v) = \bot$ if no such vertex exists).
\end{itemize}
Following the intuition above, assume that $X$ is a~low-rank subset of $V(G)$ and let $v \in V(G)$.
If $v \in X$, then $\varphi^+(v)$ specifies the twin of $v$ in $X$ with respect to $V(G) \setminus X$; and conversely, if $v \notin X$, then $\varphi^-(v)$ denotes the twin of $v$ in $V(G) \setminus X$ with respect to $X$.
In particular, $\varphi^+(v) = \bot$ implies that $v \notin X$, and $\varphi^-(v) = \bot$ implies that $v \in X$.
Also, $\varphi^+(v) \neq \bot$ implies that $\varphi^+(\varphi^+(v)) = \varphi^+(v)$ and similarly $\varphi^-(v) \neq \bot$ implies that $\varphi^-(\varphi^-(v)) = \varphi^-(v)$.

Now, given a~graph $G$ and a~tuple $\tup{a} \in V(G)^k$, we construct the~digraph $H_{\tup{a}} = (V, E_{\tup{a}})$ as follows.
We say that a~tuple $\tup{a}$ is \emph{admissible} if $\tup{a}^+$ and $\tup{a}^-$ are disjoint and in $G[\tup{a}]$, the set $\tup{a}^+$ of vertices has rank at most~$r$; the tuple is \emph{inadmissible} otherwise.
If $\tup{a}$ is inadmissible, we specify that $\vec{uv}$ is an~arc of $H_{\tup{a}}$ if at least one of the endpoints of the arc belongs to $\tup{a}$ or $uv$ is an~edge of $G$.
On the other hand, if $\tup{a}$ is admissible, we specify that $\vec{uv}$ is an~arc of $H_{\tup{a}}$ if any of the following conditions hold:
\begin{enumerate}[(i),nosep]
    \item \label{item:cond-is-param} $u \in \tup{a}^-$ or $v \in \tup{a}^+$;
    \item \label{item:cond-must-be-minus} $\varphi^+(u) = \bot$ and $v \in \tup{a}^-$;
    \item \label{item:cond-must-be-plus} $u \in \tup{a}^+$ and $\varphi^-(v) = \bot$;
    \item \label{item:cond-general} $E(u, v)\ \textrm{xor}\ [\varphi^+(u) \neq \bot\ \textrm{and}\ \varphi^-(v) \neq \bot \ \textrm{and}\ E(\varphi^+(u), \varphi^-(v))]$.
\end{enumerate}

Intuitively, condition \ref{item:cond-is-param} ensures that whenever $X$ is a~non-trivial suffix of $H_{\tup{a}}$ (i.e., a~suffix different from $\emptyset$ and $V(G)$), we necessarily have $\tup{a}^- \cap X = \emptyset$ and $\tup{a}^+ \subseteq X$.
Then conditions \ref{item:cond-must-be-minus}, \ref{item:cond-must-be-plus} ensure that $u \notin X$ whenever $\varphi^+(u) = \bot$ and $v \in X$ whenever $\varphi^-(v) = \bot$, respectively.
Finally, \ref{item:cond-general} will resolve the general case, encoding all possible implications of the form ``if some vertex $u$ belongs to $X$, then another vertex $v$ belongs to $X$ as well''.

\begin{example}
    Consider the graph $G$ from \Cref{sfig:dir-flip-example-before}, with $\tup{a}^- = (a^-_1, a^-_2)$ and $a^+ = (a^+_1, a^+_2)$.
    We have:
    \begin{center}
        \begin{tabular}{|c||c|c|c|c|c|c|c|c|}
            \hline
            $\mathbf{w}$ & $a^-_1$ & $a^-_2$ & $a^+_1$ & $a^+_2$ & $w_1$ & $w_2$ & $w_3$ & $w_4$ \\ \hline \hline
            $\mathbf{\varphi^+(w)}$ & $\bot$ & $\bot$ & $a^+_1$ & $a^+_2$ & $\bot$ & $a^+_1$ & $a^+_2$ & $a^+_2$ \\ \hline
            $\mathbf{\varphi^-(w)}$ & $a^-_1$ & $a^-_2$ & $\bot$ & $\bot$ & $a^-_1$ & $a^-_1$ & $a^-_1$ & $\bot$ \\ \hline
        \end{tabular}
    \end{center}
    The corresponding digraph $H_{\tup{a}}$ is given in \Cref{sfig:dir-flip-example-after}.

    Black arcs are of type \ref{item:cond-is-param}; these ensure that every non-trivial suffix of $H_{\tup{a}}$ contains $a^+_1, a^+_2$ and does not contain $a^-_1, a^-_2$.

    Since $\varphi^+(w_1) = \bot$ (so $w_1$ cannot belong to $X$), we additionally introduce arcs of type \ref{item:cond-must-be-minus} (blue), warranting that non-trivial suffixes of $H_{\tup{a}}$ do not contain $w_1$.

    Since $\varphi^-(w_4) = \bot$ (so $w_4$ must belong to $X$), we add arcs of type \ref{item:cond-must-be-plus} (yellow), ensuring that non-trivial suffixes of $H_{\tup{a}}$ contain $w_4$.

    Observe that $w_2w_3$ is a~non-edge in $G$, but $\varphi^+(w_2) \varphi^-(w_3)$ is an~edge in $G$.
    Therefore it cannot be that $w_2 \in X$ and $w_3 \notin X$: the statement $w_2 \in X$ implies that $\varphi^+(w_2)$ is a~twin of $w_2$ with respect to $V(G) \setminus X$, and the statement $w_3 \notin X$ implies that $\varphi^-(w_3)$ is a~twin of $w_3$ with respect to $X$.
    The (pink) arc $\vec{w_2w_3}$ in $H_{\tup{a}}$ of type \ref{item:cond-general} encodes this dependency: every suffix of $H_{\tup{a}}$ containing $w_2$ must also contain $w_3$.

    The non-trivial suffixes of $H_{\tup{a}}$ are $\{a^+_1, a^+_2, w_4\}$, $\{a^+_1, a^+_2, w_3, w_4\}$ and $\{a^+_1, a^+_2, w_2, w_3, w_4\}$.
    It can be verified that those are precisely the sets $X$ for which $\{a^+_1, a^+_2\}$ is a~representative of $X$ and $\{a^-_1, a^-_2\}$ is a~representative of $V(G) \setminus X$.
    
    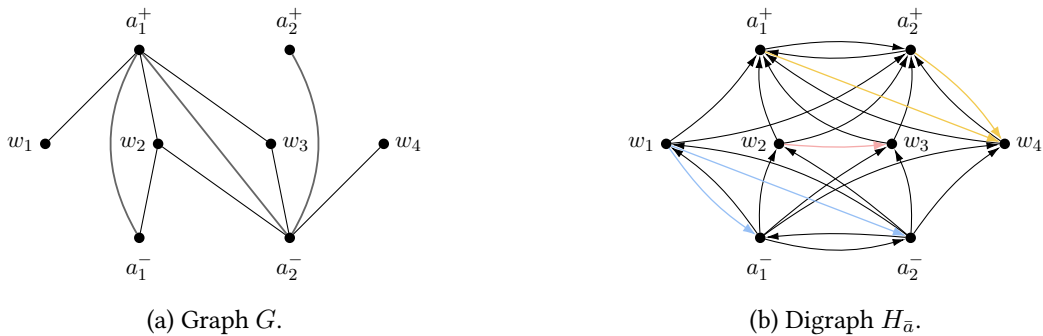
\begin{figure}[ht]
        \begin{subfigure}[t]{0.5\textwidth}
            \centering
            \begin{tikzpicture}[every node/.style={draw, circle, scale=0.8, thick,fill=black,inner sep=1.4}]
                \node[label=180:$w_1$] (w1) at (-2.25,0) {};
                \node[label=180:$w_2$] (w2) at (-0.75,0) {};
                \node[label=0:$w_3$] (w3) at (0.75,0) {};
                \node[label=0:$w_4$] (w4) at (2.25,0) {};

                \node[label=90:$a_1^+$] (a1p) at (-1,1.25) {};
                \node[label=90:$a_2^+$] (a2p) at (1,1.25) {};
                \node[label=-90:$a_1^-$] (a1m) at (-1,-1.25) {};
                \node[label=-90:$a_2^-$] (a2m) at (1,-1.25) {};

                \draw (a1p) to (w1) (a1p) to (w2) (a1p) to (w3);
                \draw (a1m) to (w2) (a2m) to (w2) (a2m) to (w3) (a2m) to (w4);
                \draw[line width=0.7pt,black!60!white] (a1p) [bend right=30] to (a1m) (a1p) [bend right=0] to (a2m) (a2p) [bend left=30] to (a2m);
            \end{tikzpicture}   
            \caption{Graph $G$.}
            \label{sfig:dir-flip-example-before}
        \end{subfigure}%
        \begin{subfigure}[t]{0.5\textwidth}
            \centering
            \begin{tikzpicture}[every node/.style={draw, circle, scale=0.8, thick,fill=black,inner sep=1.4}]
                \node[label=180:$w_1$] (w1) at (-2.25,0) {};
                \node[label=180:$w_2$] (w2) at (-0.75,0) {};
                \node[label=0:$w_3$] (w3) at (0.75,0) {};
                \node[label=0:$w_4$] (w4) at (2.25,0) {};

                \node[label=90:$a_1^+$] (a1p) at (-1,1.25) {};
                \node[label=90:$a_2^+$] (a2p) at (1,1.25) {};
                \node[label=-90:$a_1^-$] (a1m) at (-1,-1.25) {};
                \node[label=-90:$a_2^-$] (a2m) at (1,-1.25) {};

                \draw[arrows = {-Latex[width=0pt 7, length=5pt]}] (w1) [bend right= 10] to (a1p);
                \draw[arrows = {-Latex[width=0pt 7, length=5pt]}]  (w1) [bend right= 15] to (a2p);
                \draw[arrows = {-Latex[width=0pt 7, length=5pt]}] (w2) [bend left= 15] to (a1p);
                \draw[arrows = {-Latex[width=0pt 7, length=5pt]}]  (w2) [bend right= 27] to (a2p);
                \draw[arrows = {-Latex[width=0pt 7, length=5pt]}] (w3) [bend left= 27] to (a1p);
                \draw[arrows = {-Latex[width=0pt 7, length=5pt]}]  (w3) [bend right= 15] to (a2p);
                \draw[arrows = {-Latex[width=0pt 7, length=5pt]}] (w4) [bend left= 15] to (a1p);
                \draw[arrows = {-Latex[width=0pt 7, length=5pt]}]  (w4) [bend left= 10] to (a2p);
                \draw[arrows = {-Latex[width=0pt 7, length=5pt]}]  (a1p) [bend left= 10] to (a2p);
                \draw[arrows = {-Latex[width=0pt 7, length=5pt]}]  (a2p) [bend left= 10] to (a1p);

                \draw[arrows = {-Latex[width=0pt 7, length=5pt]}] (a1m) [bend right= 10] to (w1);
                \draw[arrows = {-Latex[width=0pt 7, length=5pt]}]  (a2m) [bend right= 15] to (w1);
                \draw[arrows = {-Latex[width=0pt 7, length=5pt]}] (a1m) [bend left= 15] to (w2);
                \draw[arrows = {-Latex[width=0pt 7, length=5pt]}]  (a2m) [bend right= 5] to (w2);
                \draw[arrows = {-Latex[width=0pt 7, length=5pt]}] (a1m) [bend left= 5] to (w3);
                \draw[arrows = {-Latex[width=0pt 7, length=5pt]}]  (a2m) [bend right= 15] to (w3);
                \draw[arrows = {-Latex[width=0pt 7, length=5pt]}] (a1m) [bend left= 15] to (w4);
                \draw[arrows = {-Latex[width=0pt 7, length=5pt]}]  (a2m) [bend left= 10] to (w4);
                \draw[arrows = {-Latex[width=0pt 7, length=5pt]}]  (a1m) [bend right= 15] to (a2m);
                \draw[arrows = {-Latex[width=0pt 7, length=5pt]}]  (a2m) [bend right= 5] to (a1m);

                \draw[line width=0.5pt,custom-pink,arrows = {-Latex[width=0pt 7, length=5pt]}]  (w2) [bend right= 5] to (w3);

                \draw[line width=0.5pt,custom-blue,arrows = {-Latex[width=0pt 7, length=5pt]}]  (w1) [bend right= 10] to (a1m);
                \draw[line width=0.5pt,custom-blue,arrows = {-Latex[width=0pt 7, length=5pt]}]  (w1) [bend right= 0] to (a2m);
                \draw[line width=0.5pt,custom-orange,arrows = {-Latex[width=0pt 7, length=5pt]}]  (a1p) [bend left= 0] to (w4);
                \draw[line width=0.5pt,custom-orange,arrows = {-Latex[width=0pt 7, length=5pt]}]  (a2p) [bend left= 10] to (w4);
            \end{tikzpicture}
            \caption{Digraph $H_{\tup{a}}$.}
            \label{sfig:dir-flip-example-after}
        \end{subfigure}
        \caption{Example graph $G$ and construction of the~digraph $H_{\tup{a}}$ from an~admissible tuple $\tup{a}$.}
        \label{fig:dir-flip-example}
    \end{figure}
\end{example}

We first show that $H_{\tup{a}}$ is indeed a~flip of $G$ with parameters $\tup{a}$:

\begin{lemma}
    \label{lem:low-rank-to-suffixes-exists-relation}
    There exists a~binary relation $A \subseteq \atp^{k+1} \times \atp^{k+1}$ such that, for any undirected graph $G$ and any tuple $\tup{a} \in V(G)^k$, we have $H_{\tup{a}} = G \oplus_{\tup{a}} A$.
\end{lemma}
\begin{proof}
    Let $\tau_u, \tau_v \in \atp^{k+1}$.
    We need to prove that for any undirected graph $G$, any $k$-tuple $\tup{a}$ and any pair of distinct vertices $u, v \in V(G)$ with $\atp(u, \tup{a}) = \tau_u$, $\atp(v, \tup{a}) = \tau_v$, the value $\lambda_{\tup{a}}(u, v) \coloneqq E_{\tup{a}}(u, v)\ \textrm{xor}\ E(u, v)$ depends only on the pair of types $(\tau_u, \tau_v)$.
    Hence choose $G$, $\tup{a}$, $u$ and $v$ as above.

    From $\tau_u$ (and analogously from $\tau_v$) we can deduce whether $\tup{a}$ is admissible.
    First, suppose that $\tup{a}$ is inadmissible.
    If $u \in \tup{a}$ (which can be deduced from $\tau_u$), then we have $(u, v) \in E_{\tup{a}}$, so the value $\lambda_{\tup{a}}(u, v) = E_{\tup{a}}(u, v) \ \textrm{xor}\ E(u, v) = \neg E(u, v)$ is uniquely implied by $\tau_v$ (since $u \in \tup{a}$ and the atomic type $\tau_v$ determines $N(v) \cap \tup{a}$).
    Symmetrically, if $v \in \tup{a}$, then $\lambda_{\tup{a}}(u, v)$ can be deduced from $\tau_u$.
    On the other hand, if $u, v \notin \tup{a}$, then $\lambda_{\tup{a}}(u, v)$ is simply false, so it is enough to specify that $(\tau_u, \tau_v) \notin A$.

    Now consider the case when $\tup{a}$ is admissible.
    The satisfaction of any of the conditions \ref{item:cond-is-param}, \ref{item:cond-must-be-minus}, \ref{item:cond-must-be-plus} can be verified by examining the atomic types $\tau_u$ and $\tau_v$.
    In each of these cases, we have $\lambda_{\tup{a}}(u, v) = \neg E(u, v)$, and the value $E(u, v)$ is uniquely determined by $\tau_u$ (if $v \in \tup{a}$) or by $\tau_v$ (if $u \in \tup{a}$).
    On the other hand, if none of these cases hold, then the value $\lambda_{\tup{a}}(u, v) = [\varphi^+(u) \neq \bot\ \textrm{and}\ \varphi^-(v) \neq \bot\ \textrm{and}\ E(\varphi^+(u), \varphi^-(v))]$ can be inferred from the pair $(\tau_u, \tau_v)$.
\end{proof}

It remains to show that the family of subsets of $V(G)$ of rank at most $r$ is equivalent to the family of suffixes of the graphs $H_{\tup{a}}$.
We prove each implication of this equivalence separately. First we  argue that each suffix of $H_{\tup{a}}$ indeed has small rank.

\begin{lemma}
    \label{lem:low-rank-to-suffixes-rtl}
    Let $G$ be an~undirected graph, $\tup{a} \in V(G)^k$ and $X \in \Suffixes(H_{\tup{a}})$.
    Then $\rk(X) \leq r$.
\end{lemma}
\begin{proof}
    If $\tup{a}$ is inadmissible, then $H_{\tup{a}}$ is undirected and every vertex $v \in \tup{a}$ is a~universal vertex of $H_{\tup{a}}$.
    In particular, $\Suffixes(H_{\tup{a}}) = \{\emptyset, V(G)\}$ and the lemma holds.
    From now on we assume that $\tup{a}$ is admissible and that $X \notin \{\emptyset, V(G)\}$.

    By condition \ref{item:cond-is-param} and the fact that $X$ is a~suffix of $H_{\tup{a}}$, we infer that $\tup{a}^+ \subseteq X$; this is because we can choose an~arbitrary vertex $w \in X$ and then for any $v \in \tup{a}^+$ we have $\vec{wv} \in E_{\tup{a}}$ and so $v \in X$.
    We symmetrically get $\tup{a}^- \cap X = \emptyset$.
    Then, by \ref{item:cond-must-be-minus}, if some $u \in V(G)$ satisfies $\varphi^+(u) = \bot$, then $u \notin X$ (since we can choose some $v \in \tup{a}^-$, and then $u \notin X$ follows from $\vec{uv} \in E_{\tup{a}}$ and $v \notin X$).
    Similarly, from \ref{item:cond-must-be-plus} we find that $\varphi^-(v) = \bot$ implies that $v \in X$.
    In particular, no vertex $w \in V(G)$ has $\varphi^-(w) = \varphi^+(w) = \bot$.

    Pick now any $u \in X$ and $v \notin X$.
    Then $\vec{uv} \notin E_{\tup{a}}$, so none of the conditions \ref{item:cond-is-param}, \ref{item:cond-must-be-minus}, \ref{item:cond-must-be-plus} applies, and in particular $E(u, v) = E(\varphi^+(u),\, \varphi^-(v))$ by the statement of condition \ref{item:cond-general} and the facts that $\varphi^+(u) \neq \bot$ and $\varphi^-(v) \neq \bot$.
    Similarly, we have $\vec{\varphi^+(u)v} \notin E_{\tup{a}}$, so $E(\varphi^+(u), v) = E(\varphi^+(u), \varphi^-(v))$ by \ref{item:cond-general} and $\varphi^+(\varphi^+(u)) = \varphi^+(u)$.
    As $u \in X$, $v \notin X$ were arbitrary, this implies that $N(u) \setminus X = N(\varphi^+(u)) \setminus X$ for every $u \in X$.
    Symmetrically, we prove $N(v) \cap X = N(\varphi^-(v)) \cap X$ for each $v \notin X$.
    Therefore, in the adjacency matrix $\Adj_G[X, \tup X]$, where $\tup X=V(G)\setminus X$, the rows corresponding to vertices $u$ and $\varphi^+(u)$ coincide, as do the columns corresponding to vertices $v$ and $\varphi^-(v)$.
    We conclude that
    \[ \rk(X) = \rkk(\Adj_G[X, \tup X]) = \rkk(\Adj_G[\tup{a}^+, \tup{a}^-]) \leq r, \]
    where the last inequality is due to the admissibility of $\tup{a}$.
\end{proof}

We then show the converse implication.

\begin{lemma}
    \label{lem:low-rank-to-suffixes-ltr}
    Let $G$ be an~undirected graph and pick $X \subseteq V(G)$ with $\rk(X) \leq r$.
    Then there is $\tup{a} \in V(G)^k$ such that $X \in \Suffixes(H_{\tup{a}})$.
\end{lemma}
\begin{proof}
    The lemma is trivial for $X \in \{\emptyset, V(G)\}$, so assume otherwise.
    By \Cref{prop:small-representative}, there exist representatives of $X$ and $V(G) \setminus X$ in $G$ of size at most $2^r$ (and both are nonempty); call them $R_+$ and $R_-$, respectively.
    Let $\tup{a}$ be a~$k$-tuple of parameters in which $\tup{a}^+$ (resp., $\tup{a}^-$) contains precisely the vertices of $R_+$ (resp., $R_-$), each at least once.
    We claim that $X \in \Suffixes(H_{\tup{a}})$.

    Naturally, $\tup{a}$ is admissible.
    Choose now two vertices $u \in X$, $v \notin X$ and assume for the sake of contradiction that $\vec{uv} \in E_{\tup{a}}$.
    Condition \ref{item:cond-is-param} cannot apply since $u \in X$, $v \notin X$; condition \ref{item:cond-must-be-minus} also cannot apply as $u$ has a~twin in $X$ with respect to $V(G) \setminus X$ and thus $\varphi^+(u) \neq \bot$; and condition \ref{item:cond-must-be-plus} cannot apply by a~symmetric argument for $v$.
    Now, $\varphi^+(u) \in X$ is a~twin of $u$ with respect to $\tup{a}^-$, and so $\varphi^+(u)$ and $u$ are twins with respect to $V(G) \setminus X$.
    Similarly, $\varphi^-(v) \in V(G) \setminus X$ and $v$ are twins with respect to $X$, implying that $E(u, v) = E(\varphi^+(u), v) = E(\varphi^+(u), \varphi^-(v))$.
    But this excludes condition \ref{item:cond-general}, since $\varphi^+(u) \neq \bot$, $\varphi^-(v) \neq \bot$, and $E(u, v) = E(\varphi^+(u), \varphi^-(v))$.
    The contradiction proves that $\vec{uv} \notin E_{\tup{a}}$.
\end{proof}

\Cref{lem:low-rank-to-suffixes-exists-relation,lem:low-rank-to-suffixes-rtl,lem:low-rank-to-suffixes-ltr} complete the proof of \Cref{lem:low-rank-to-suffixes}.

\subsubsection{Parameterizing suffixes of flips}

We now prove the following result, which essentially shows that the family of all suffixes of a~directed graph can be described in a~concise way, provided that the graph is a~definable flip of an~undirected graph.

\begin{lemma}
    \label{lem:suffixes-to-spans}
    Let $k \in \N$ and $A \subseteq \atp^{k+1} \times \atp^{k+1}$.
    Then there exist $\ell \in \N$ and formulas $\varphi_+$, $\varphi_-$, $\varphi_\sim$ of flip-reachability logic with the following properties:
    \begin{itemize}[nosep]
        \item For any undirected graph $G$ and any $(k+\ell)$-tuple $\tup{a}\tup{b} \in V(G)^{k+\ell}$ with $|\tup{a}| = k$, $|\tup{b}| = \ell$,   the quadruple $\varphi_+, \varphi_-, \varphi_\sim, \tup{a}\tup{b}$ defines an~$\tup{a}$-uniform seed of $G$.
        \item For any undirected graph $G$ and any $k$-tuple $\tup{a} \in V(G)^k$,
        \begin{equation}
            \label{eq:suffixes-eq-span}
            \Suffixes(G \oplus_{\tup{a}} A) = \bigcup_{\tup{b} \in V(G)^\ell} \Span(\Seed(\varphi_+, \varphi_-, \varphi_{\sim}, \tup{a}\tup{b})).
        \end{equation}
    \end{itemize}
\end{lemma}

The remainder of this section is devoted to the proof of \Cref{lem:suffixes-to-spans}.
Observe here that \Cref{lem:low-rank-to-suffixes,lem:suffixes-to-spans} immediately imply the Low Rank Structure Theorem (\Cref{thm:low-rank-structure}).

Let us start with the intuition.
Suppose we are given a~digraph $H$ --- a~flip of an~undirected graph $G$ defined by some $k$ parameters --- and we wish to describe the structure of suffixes of $H$.
Observe first that each suffix is necessarily the~union of some collection of strongly connected components of $H$; in particular, each strongly connected component either entirely belongs to a~suffix of $H$, or is disjoint from the suffix.
In fact, the formula $\varphi_\sim$ of flip-reachability logic will precisely encode the equivalence relation selecting pairs of vertices that belong to the same strongly connected component of $H$.
This way, in each seed $(X_+, X_-, \Xc)$ that we shall construct, both $X_+$ and $X_-$ will be consist of the union of some strongly connected components of $H$, while $\Xc$ will be the partition of $V(G) \setminus (X_+ \cup X_-)$ into strongly connected components of $H$. Thus, one may think of $X_+$ as of the union of those components that must be included in the suffix, and similarly $X_-$ is comprised of those components that cannot be included, while the components of $\Xc$ can be freely included or not.

However, we must additionally ensure another property of $\Xc$: namely, $\Xc$ must be a~collection of \emph{incomparable} strongly connected components of $H$.
That is, no component of $\Xc$ may be reachable in $H$ from any other component of $\Xc$.
In other words, $\Xc$ must form an~\emph{antichain} in the partially ordered set (\emph{poset}) $(\Scc, \leq)$, which is formed from $H$ by setting $\Scc$ to be the family of all strongly connected components of $H$ and specifying that $C_1 \leq C_2$ for $C_1, C_2 \in \Scc$ whenever there exists a~directed path in $H$ from any vertex of $C_1$ to any vertex of $C_2$.

This suggests the following strategy: for every inclusionwise maximal antichain $\Xc$ of $(\Scc, \leq)$, we will form a~seed $(X_+, X_-, \Xc)$ by choosing $X_+$ as (the union of) the collection of components of $H$ that are strictly greater than some component of $\Xc$ in the poset, and $X_-$ as (the union of) the collection of components of $H$ that are strictly smaller.
Then, as it will turn out, $\Suffixes(H)$ will be the union of the spans of all the seeds $(X_+, X_-, \Xc)$ formed in this process.
One challenge remains: it is not clear how to parameterize each such seed  with a~bounded number of vertices.
However, this turns out to be actually possible, and this observation constitutes the combinatorial heart of our proof. We will prove that in each relevant seed $(X_+, X_-, \Xc)$, we can find a~subfamily $\Xc^\star \subseteq \Xc$ of \emph{bounded} size (in terms of a~function in $k$) such that $X_+$ and $X_-$ still remain the unions of the collections of components of $H$ that are, respectively, strictly greater or strictly smaller than any element of $\Xc^\star$ in the poset $(\Scc, \leq)$.
This ultimately allows us to parameterize each seed in terms of a~bounded-length tuple $\tup{b}$ of vertices of $H$: for each component of $\Xc^\star$, include in $\tup{b}$ an~arbitrary vertex of this component.
Then $X_+$ will contain vertices of $G$ that can be reached from $\tup{b}$ in $H$ but cannot reach any vertex of $\tup{b}$, and similarly, $X_-$ will contain vertices of $G$ that can reach in $H$ some vertex of $\tup{b}$ yet cannot be reached from any vertex of $\tup{b}$.
We are ready to proceed to formal details.

\newcommand{\IsConsistent}{\mathsf{Consistent}}

Fix $k \in \N$ and $A \subseteq \atp^{k+1} \times \atp^{k+1}$ as in the statement of the lemma.
For $\ell \in \N$, whose exact value we will determine later, we define the formulas $\varphi_+(z, \tup{x}\tup{y})$, $\varphi_-(z, \tup{x}\tup{y})$, $\varphi_\sim(z_1, z_2, \tup{x}\tup{y})$ with $|\tup{x}| = k$, $|\tup{y}| = \ell$ (so that $\tup{x}$ is the tuple of variables representing the parameters $\tup{a}$ of the flip and $\tup{y}$ is the tuple of variables representing vertices $\tup{b}$ of the antichain $\Xc^\star$ as in the description above), as follows.
Let $z_1 \leq_{\tup{x}} z_2$ be the shorthand for $\flipreach_A(z_1, z_2, \tup{x})$; that is, $z_1 \leq_{\tup{x}} z_2$ holds iff $z_2$ is reachable from $z_1$ in the flip $G \oplus_A \tup{x}$.
We also define $z_1 <_{\tup{x}} z_2$ as the shorthand for $(z_1 \leq_{\tup{x}} z_2) \wedge \neg (z_2 \leq_{\tup{x}} z_1)$.
Observe that $\leq_{\tup{x}}$ and $<_{\tup{x}}$ are transitive in the following sense: if $z_1 \leq_{\tup{x}} z_2$ and $z_2 \leq_{\tup{x}} z_3$, then also $z_1 \leq_{\tup{x}} z_3$; and if at least one of the two inequalities is actually strict (i.e., if it additionally holds that $z_1 <_{\tup{x}} z_2$ or $z_2 <_{\tup{x}} z_3$), then also $z_1 <_{\tup{x}} z_3$.
Finally, define the following formulas:
\begin{align*}
    \varphi_\sim(z_1, z_2, \tup{x}\tup{y}) &= (z_1 \leq_{\tup{x}} z_2) \wedge (z_2 \leq_{\tup{x}} z_1), \\
    \widehat{\varphi}_+(z, \tup{x}\tup{y}) &= \bigvee_{i \in [\ell]} y_i <_{\tup{x}} z, \\
    \widehat{\varphi}_-(z, \tup{x}\tup{y}) &= \bigvee_{i \in [\ell]} z <_{\tup{x}} y_i, \\
    \widehat{\varphi}_\Xc(z, \tup{x}\tup{y}) &= \neg\widehat{\varphi}_+(z, \tup{x}\tup{y}) \wedge \neg\widehat{\varphi}_-(z, \tup{x}\tup{y}), \\
    \IsConsistent(\tup{x}\tup{y}) &= \bigwedge_{i \in [\ell]} \widehat{\varphi}_\Xc(y_i, \tup{x}\tup{y}) \wedge \neg\exists_{z_1, z_2} \left( \widehat{\varphi}_\Xc(z_1, \tup{x}\tup{y}) \wedge \widehat{\varphi}_\Xc(z_2, \tup{x}\tup{y}) \wedge z_1 <_{\tup{x}} z_2 \right), \\
    \varphi_+(z, \tup{x}\tup{y}) &= \IsConsistent(\tup{x}\tup{y}) \wedge \widehat{\varphi}_+(z, \tup{x}\tup{y}), \\
    \varphi_-(z, \tup{x}\tup{y}) &= \neg\IsConsistent(\tup{x}\tup{y}) \vee \widehat{\varphi}_-(z, \tup{x}\tup{y}).
\end{align*}
Here, $\varphi_\sim(v_1, v_2, \tup{a}\tup{b})$ tests whether two vertices $v_1, v_2$ are in the same strongly connected component of $H = G \oplus_A \tup{a}$ (in which case they belong to the same set of $\Xc$ for a~seed $(X_+, X_-, \Xc)$ with $v_1, v_2 \in \bigcup \Xc$).
Then the formulas $\widehat{\varphi}_+, \widehat{\varphi}_-$ tentatively determine the sets $X_+, X_-$ of the seed $(X_+, X_-, \Xc)$, given a~sample $\tup{b}$ of vertices chosen from $\bigcup \Xc$: we declare that $X_+$ is the set of vertices reachable in $H$ from $\tup{b}$ but which cannot reach $\tup{b}$, and $X_-$ are the vertices that can reach $\tup{b}$, but are unreachable from $\tup{b}$.
Then $\widehat{\varphi}_\Xc$ determines the set of vertices in $\bigcup \Xc$ as all vertices outside of $X_+$ and $X_-$; we consider $\bigcup \Xc$ to be partitioned into sets $\Xc$ according to $\varphi_\sim$.
We will say that such a~tuple $(X_+, X_-, \Xc)$ is \emph{generated} by $\tup{b}$.

However, not every choice of $\tup{b}$ generates a~valid seed: we have to first ensure that each element of $\tup{b}$ \emph{actually belongs} to $\bigcup \Xc$, and that $\Xc$ does not contain two comparable strongly connected components.
Let us define that $\tup{b}$ is \emph{consistent} if it passes both of these checks, and observe that $\IsConsistent(\tup{a}\tup{b})$ verifies the consistency of $\tup{b}$ according to our definition.
The final contents of the seed identified by $\varphi_+, \varphi_-, \varphi_\sim, \tup{a}\tup{b}$ depends on whether $\tup{b}$ is consistent: if this is the case, we simply take $(X_+, X_-, \Xc)$ as the seed (and in this case we have $\varphi_+ = \widehat{\varphi}_+$ and $\varphi_- = \widehat{\varphi}_-$), and otherwise we use the trivial seed $(\emptyset, V(G), \emptyset)$ with span $\{\emptyset\}$ (in which case we define $\varphi_+$ to always fail and $\varphi_-$ to always hold).

We now prove several properties of our formulas across several lemmas. Together, these will amount to the proof of \Cref{lem:suffixes-to-spans}.

\begin{lemma}\label{lem:def-to-seed}
    Let $G$ be an~undirected graph and let $\tup{a} \in V(G)^k$, $\tup{b} \in V(G)^\ell$.
    Then $\Seed(\varphi_+, \varphi_-, \varphi_\sim, \tup{a}\tup{b})$ is a~seed that is $\tup{a}$-uniform in $G$.
\end{lemma}
\begin{proof}
    First we show that $\Seed(\varphi_+, \varphi_-, \varphi_\sim, \tup{a}\tup{b})$ indeed forms a~seed $(X_+, X_-, \Xc)$.
    Recall that $X_+ = \{v \in V(G) \,\mid\, G \models \varphi_+(v, \tup{a})\}$ and $X_- = \{v \in V(G) \,\mid\, G \models \varphi_-(v, \tup{a})\}$.
    Since $\varphi_\sim$ naturally defines an~equivalence relation on $V(G)$, the only non-trivial part of the proof is that $X_+$ and $X_-$ are disjoint.
    This is obviously the case when $\tup{b}$ is inconsistent, so suppose otherwise.
    Then a~vertex $v \in X_+ \cap X_-$ would witness that $b_i <_{\tup{x}} v$ for some $i \in [\ell]$, as well as $v <_{\tup{x}} b_j$ for some $j \in [\ell]$.
    Therefore $b_i <_{\tup{x}} b_j$; but this contradicts the consistency of $\tup{b}$.

    Now suppose that the seed is not $\tup{a}$-uniform in $G$.
    Again, assume that $\tup{b}$ is consistent as the inconsistent case is trivial.
    Recall that $\Xc$ is the set of equivalence classes of $V(G) \setminus (X_+ \cup X_-)$ defined with respect to $\varphi_\sim$ (and by the definitions of $\varphi_+, \varphi_-, \varphi_\sim$, we have that $\Xc$ is the~union of a collection of strongly connected components of $H$).
    Then we have three vertices $u_1, u_2, v \in V(G)$ such that $u_1, u_2 \in \bigcup \Xc$, $\atp(u_1, \tup{a}) = \atp(u_2, \tup{a})$, $v \notin \Xc(u_1) \cup \Xc(u_2)$ and, without loss of generality, $u_1v \in E(G)$ and $u_2v \notin E(G)$.
    Since $H = G \oplus_{\tup{a}} A$, we have by definition:
    \begin{align*}
        \vec{u_1v} \in E(H) \,&\Longleftrightarrow\, \neg [ (\atp(u_1, \tup{a}), \atp(v, \tup{a})) \in A ], \\
        \vec{u_2v} \in E(H) \,&\Longleftrightarrow\, (\atp(u_2, \tup{a}), \atp(v, \tup{a})) \in A.
    \end{align*}
    Since $\atp(u_1, \tup{a}) = \atp(u_2, \tup{a})$, precisely one of the arcs $\vec{u_1v}$, $\vec{u_2v}$ belongs to $H$.
    Following a~symmetric argument, we also find that exactly one of the arcs $\vec{vu_1}$, $\vec{vu_2}$ belongs to $H$.
    But then one of the following cases holds:
    \begin{itemize}[nosep]
        \item $\vec{u_iv}, \vec{vu_i} \in E(H)$ for some $i \in \{1,2\}$.
        But then $v \in \Xc(u_i)$ --- a~contradiction.
        \item $\vec{u_iv}, \vec{vu_j} \in E(H)$ for some $\{i, j\} = \{1, 2\}$.
        Since $v \notin \Xc(u_i)$, we necessarily have that $u_i <_{\tup{x}} v$; similarly, we have that $v <_{\tup{x}} u_j$.
        But then $u_i <_{\tup{x}} u_j$, and since $u_i, u_j \in \bigcup \Xc$, we infer that $\widehat{\varphi}_\Xc(u_i, \tup{a}\tup{b}) \wedge \widehat{\varphi}_\Xc(u_j, \tup{a}\tup{b}) \wedge (u_i <_{\tup{x}} u_j)$. This contradicts the consistency of $\tup{b}$.
    \end{itemize}
    Therefore, the considered seed is $\tup{a}$-uniform.
\end{proof}

Next, we verify both inclusions of equation \eqref{eq:suffixes-eq-span}; we begin with the simpler one.

\begin{lemma}
    \label{lem:param-to-suffix}
    Let $G$ be an~undirected graph, $\tup{a} \in V(G)^k$, and $\tup{b} \in V(G)^\ell$.
    Define the seed $(X_+, X_-, \Xc) = \Seed(\varphi_+, \varphi_-, \varphi_\sim, \tup{a}\tup{b})$ and consider $X \in \Span(X_+, X_-, \Xc)$.
    Then $X$ is a~suffix of $G \oplus_{\tup{a}} A$.
\end{lemma}
\begin{proof}
    Assume $X \neq \emptyset$ as otherwise there is nothing to prove.
    Then $\tup{b}$ must be consistent.
    Recall that $X = X_+ \cup \bigcup \Yc$ for some $\Yc \subseteq \Xc$ and so $V(G) \setminus X = X_- \cup \bigcup (\Xc \setminus \Yc)$.
    For the sake of contradiction assume that we have some $u \in X$ and $v \notin X$ with $\vec{uv} \in E(G \oplus_{\tup{a}} A)$.
    If $u \in X_+$, then $b_i <_{\tup{a}} u$ for some $i \in [\ell]$ and therefore $b_i <_{\tup{a}} v$ from $u \leq_{\tup{a}} v$, so $v \in X_+$ as well --- a~contradiction.
    Similarly, if $v \in X_-$, then $v <_{\tup{a}} b_j$ for some $j \in [\ell]$ and so also $u <_{\tup{a}} b_j$ from $u \leq_{\tup{a}} v$, implying $u \in X_-$ --- again a~contradiction.
    Therefore, $u \in \bigcup \Yc$ and $v \in \bigcup (\Xc \setminus \Yc)$.
    By the consistency of $\tup{b}$, it cannot be that $u <_{\tup{a}} v$.
    From our assumption that $u \leq_{\tup{a}} v$ we infer that $v \leq_{\tup{a}} u$, so $u$ and $v$ are in the same strongly connected component of $G \oplus_{\tup{a}} A$ and in particular in the same set of $\Xc$.
    But this contradicts the facts that $u \in \bigcup \Yc$ and $v \in \bigcup (\Xc \setminus \Yc)$.
    Hence $X$ must be a~suffix of $G \oplus_{\tup{a}} A$.
\end{proof}

The next lemma provides the opposite direction. In the proof we locate a~succinctly parameterizable seed containing a~given suffix of $G \oplus_{\tup{a}} A$. This is the core argument of the whole approach.

\begin{lemma}
    \label{lem:suffix-to-param}
    There exists $\ell \in \N$, depending only on $k$, with the following property.
    Let $G$ be an~undirected graph, $\tup{a} \in V(G)^k$, and $X$ be a~suffix of $G \oplus_{\tup{a}} A$.
    Then there exists $\tup{b} \in V(G)^\ell$ such that $X \in \Span(\varphi_+, \varphi_-, \varphi_\sim, \tup{a}\tup{b})$.
\end{lemma}
\begin{proof}
    Set $\ell \coloneqq 2 \cdot |\atp^{k+1}|^2$.
    For the course of the proof, fix $G$, $\tup{a}$ and $X$ as in the statement of the lemma.
    Let also $H = G \oplus_{\tup{a}} A$, and define $(\Scc, \leq)$ to be the poset of strongly connected components of $H$.
    Note that for any two vertices $u,v$ of $G$, $u \leq_{\tup{a}} v$ is equivalent to $\Scc(u) \leq \Scc(v)$ in the poset, and $u <_{\tup{a}} v$ is equivalent to $\Scc(u) < \Scc(v)$ in the poset, where $\Scc(w)$ denotes the strongly connected component that contains $w$. For brevity, a {\em{suffix}} of $(\Scc,\leq)$ is an upward-closed subset of $\Scc$ (i.e. a set $\Tt$ such that $C\leq D$ and $C\in \Tt$ implies $D\in \Tt$), and a {\em{prefix}} of $(\Scc,\leq)$ is a downward-closed subset of $\Scc$.

    Call a~seed $(X_+, X_-, \Xc)$ of $G$ \emph{splendid} if it satisfies the following properties:
    \begin{enumerate}[(i), nosep]
        \item $X \in \Span(X_+, X_-, \Xc)$,
        \item \label{item:suffix-to-param-xp-suffix} $X_+$ is the union of a~suffix of $(\Scc, \leq)$,
        \item $X_-$ is the union of a~prefix of $(\Scc, \leq)$, and
        \item \label{item:suffix-to-param-xc-antichain} $\Xc$ is an~antichain of $(\Scc, \leq)$.
    \end{enumerate}
    Since $X$ is a~suffix of $H$, both $X$ and $V(G) \setminus X$ are the unions of collections of strongly connected components of $H$; so $(X, V(G) \setminus X, \emptyset)$ is a~splendid seed.
    Choose $(X_+, X_-, \Xc)$ as any splendid seed that maximizes $|\Xc|$.
    In the following series of claims we will show that there exists $\tup{b} \in V(G)^\ell$ such that $(X_+, X_-, \Xc) = \Seed(\varphi_+, \varphi_-, \varphi_\sim, \tup{a}\tup{b})$.
    We begin by examining the relationship between $X_+$ and~$\bigcup \Xc$.

    \begin{claim}
        \label{cl:indset-union-is-covering}
        For every $v \in X_+$ there exists $u \in \bigcup \Xc$ such that $u <_{\tup{a}} v$.
    \end{claim}
    \begin{claimproof}[Proof of the claim]
        Suppose otherwise.
        Let then $P$ be the set of vertices $v \in X_+$ violating this condition: \[P \coloneqq  \left\{v \in X_+ \,\mid\, \textrm{there is no }u\in \bigcup \Xc\textrm{ such that } u <_{\tup{a}} v\right\}.\]
        $P$ is then a~(non-empty) union of some family $\Pc \subseteq \Scc$ of strongly connected components of $H$: if there existed two vertices $v_1 \in P, v_2 \notin P$ belonging the same strongly connected component of $H$, then $v_1\in X_+$ would entail $v_2\in X_+$, and $x_2\notin P$ would entail the existence of $u \in \bigcup \Xc$ such that $u <_{\tup{a}} v_2$.
        But from $v_2 \leq_{\tup{a}} v_1$ it also follows that $u <_{\tup{a}} v_1$, implying that $v_1 \notin P$ as well --- a~contradiction.

        Let $C \in \Pc$ be a~$\leq$-minimal element of the poset $(\Pc, \leq)$.
        We will now show that $(X_+ \setminus C, X_-, \Xc \cup \{C\})$ is splendid, thus contradicting the assumption about the maximality of $(X_+, X_-, \Xc)$.
        The only non-trivial checks are \ref{item:suffix-to-param-xp-suffix} and \ref{item:suffix-to-param-xc-antichain}.
        First assume that \ref{item:suffix-to-param-xp-suffix} is violated; then we have another component $C' \in \Scc$ such that $C' \subseteq X_+$ and $C' < C$.
        In particular, choosing any $v \in C$ and $v' \in C'$, we have $v' <_{\tup{a}} v$.
        Then $v' \in P$: otherwise, we would have some $u \in \bigcup \Xc$ such that $u <_{\tup{a}} v'$, hence $u <_{\tup{a}} v$ and so $v \notin P$ --- a~contradiction with $C \subseteq P$.
        Hence $C' \in \Pc$, but this contradicts the choice of $C \in \Pc$ as a~$\leq$-minimal element of~$(\Pc, \leq)$.
        
        For \ref{item:suffix-to-param-xc-antichain}, assume that $\Xc \cup \{C\}$ is not an~antichain of $(\Scc, \leq)$.
        Remembering that $\Xc$ is an~antichain of $(\Scc, \leq)$, we see that $C$ is comparable in $(\Scc, \leq)$ with some other component $D \in \Xc$.
        We verify two possible~cases:
        \begin{itemize}[nosep]
            \item If $C < D$, then we have a~direct contradiction with the assumption that $X_+$ is a~suffix of $(\Scc, \leq)$: this is the case since $C    \subseteq X_+$ and $D \in \Xc$ is disjoint from $X_+$.
            \item If $D < C$, then there exist $u \in D$, $v \in C$ such that $u <_{\tup{a}} v$.
                But this, by the definition of $P$ and the fact that $D \in \Xc$, implies that $v \notin P$, contradicting $C \subseteq P$.
        \end{itemize}
        The contradiction concludes the proof.
    \end{claimproof}

    Call a~set $B_+ \subseteq \bigcup \Xc$ an \emph{$X_+$-cover} if for every $v \in X_+$ there exists $u \in B_+$ such that $u <_{\tup{a}} v$.
    By \Cref{cl:indset-union-is-covering}, $\bigcup \Xc$ itself is an $X_+$-cover.
    Let now $B_+$ be any $X_+$-cover of minimum cardinality.

    \begin{claim}
        \label{cl:small-upper-covering}
        $|B_+| \leq |\atp^{k+1}|^2$.
    \end{claim}
    \begin{claimproof}[Proof of the claim]
        Again suppose otherwise.
        Let $p \coloneqq |\atp^{k+1}|^2 + 1$ and pick any $p$ distinct vertices $u_1, \ldots, u_p \in B_+$.
        For every $i \in [p]$, the set $B_+ \setminus \{u_i\}$ is not an $X_+$-cover, so there exists a vertex $v_i \in X_+$ such that $u_i <_{\tup{a}} v_i$ and $\neg (u <_{\tup{a}} v_i)$ for all $u \in B_+ \setminus \{u_i\}$.
        In particular, $\neg (u_j <_{\tup{a}} v_i)$ for all $j \in [p] \setminus \{i\}$.

\begin{figure}
 \centering

		\includegraphics[page=6,scale=0.25]{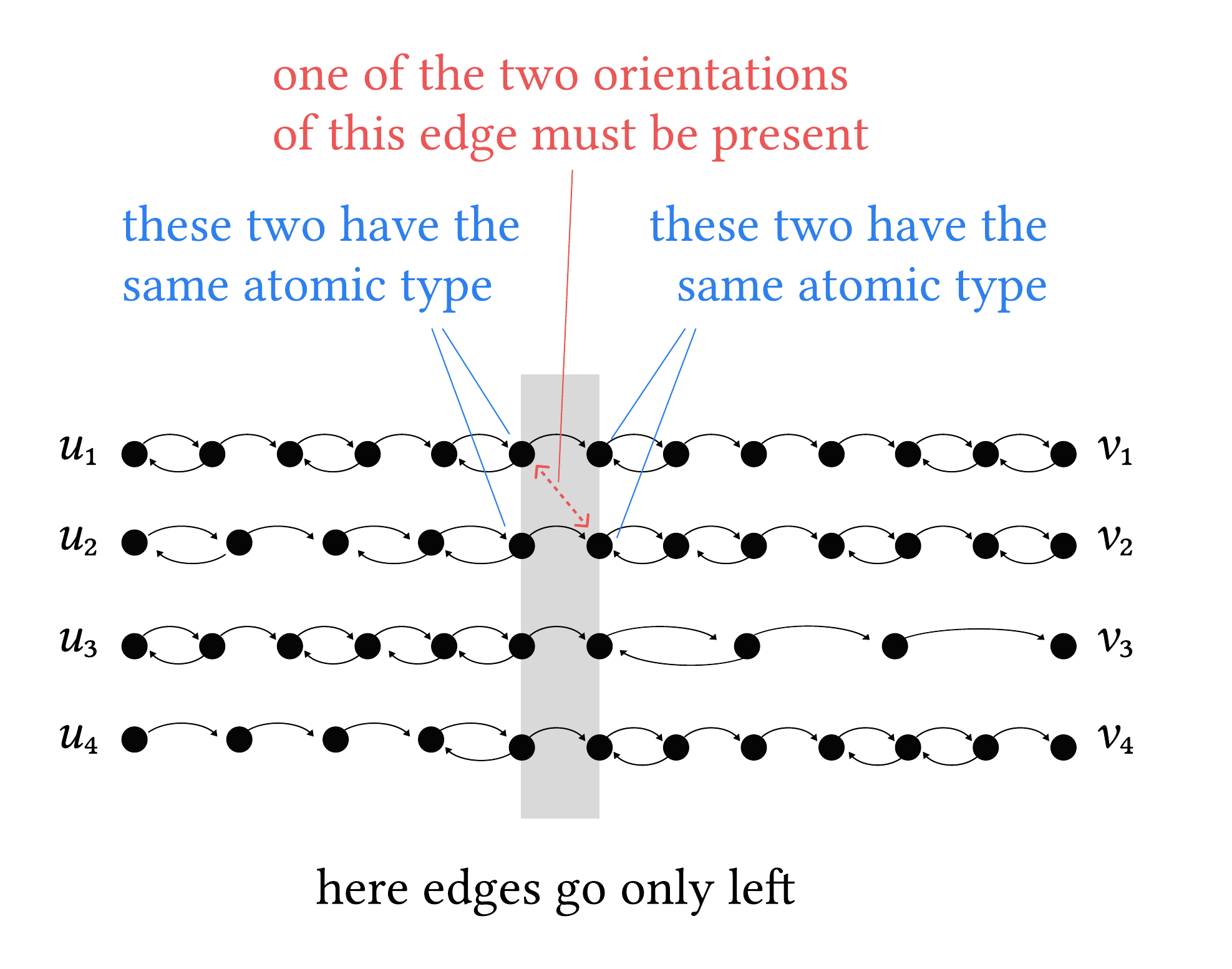}
 \caption{Situation in the proof of \Cref{cl:small-upper-covering}.}\label{fig:matching}
\end{figure}

        For $i \in [p]$ let $P_i$ be a~path from $u_i$ to $v_i$ in $H$ (such a path exists since $u_i <_{\tup{a}} v_i$).
        Since $u_i <_{\tup{a}} v_i$, there is no directed path from $v_i$ to $u_i$. In~particular reversing the directions of all arcs in $P_i$ produces a~sequence of arcs that is \emph{not} a~subpath of $H$.
        Therefore, $P_i$ contains two consecutive vertices $u'_i$, $v'_i$ such that $\vec{u'_i v'_i} \in E(H)$ and $\vec{v'_i u'_i} \notin E(H)$.

        Since $p > |\atp^{k+1}|^2$, we may assume, without loss of generality, that $\atp(u'_1, \tup{a}) = \atp(u'_2, \tup{a})$ and $\atp(v'_1, \tup{a}) = \atp(v'_2, \tup{a})$.
        Let us now define two boolean variables $\lambda_{uv}, \lambda_{vu} \in \{0, 1\}$:
        \[ \begin{split}
            \lambda_{uv} &= [(\atp(u'_1, \tup{a}), \atp(v'_1, \tup{a})) \in A], \\
            \lambda_{vu} &= [(\atp(v'_1, \tup{a}), \atp(u'_1, \tup{a})) \in A].
        \end{split} \]
        By the definition of a~flip, for all $i, j \in [2]$, we have
        \[ \begin{split}
            \vec{u'_i v'_j} \in E(H) \,&\Longleftrightarrow\, (u'_i v'_j \in E(G)) \,\textrm{xor}\, \lambda_{uv}, \\
            \vec{v'_j u'_i} \in E(H) \,&\Longleftrightarrow\, (u'_i v'_j \in E(G)) \,\textrm{xor}\, \lambda_{vu}.
        \end{split} \]
        Since $\vec{u'_1 v'_1} \in E(H)$ and $\vec{v'_1 u'_1} \notin E(H)$, we must have $\lambda_{uv} \,\textrm{xor}\, \lambda_{vu} = 1$.
        But then precisely one of $\vec{u'_1 v'_2}$, $\vec{v'_2 u'_1}$ is an~arc of $H$.
        If $\vec{u'_1 v'_2} \in E(H)$, then $u_1$ can reach $v_2$ in $H$ by following $P_1$ from $u_1$ to $u'_1$, then taking the~arc $\vec{u'_1 v'_2}$, and finally following $P_2$ to $v_2$.
        As also $\neg (u_1 <_{\tup{a}} v_2)$, we conclude that $u_1$ and $v_2$ must be in the same strongly connected component of $H$; in particular, $v_2 \leq_{\tup{a}} u_1$.
        But since $u_1 <_{\tup{a}} v_1$ and $u_2 <_{\tup{a}} v_2$, we find that $u_1 <_{\tup{a}} v_2$ --- a~contradiction.
        A~symmetric argument shows a~contradiction in the case when~$\vec{v'_2 u'_1} \in E(H)$.
    \end{claimproof}

    Analogously, we call a~set $B_- \subseteq \bigcup \Xc$ an \emph{$X_-$-cover} if for every $v \in X_-$ there exists $u \in B_-$ such that $v <_{\tup{a}} u$.
    Arguments symmetric to those in \Cref{cl:indset-union-is-covering,cl:small-upper-covering} show that:

    \begin{claim}
        \label{cl:small-lower-covering}
        There exists an~$X_-$-cover $B_-$ of cardinality at most $|\atp^{k+1}|^2$.
    \end{claim}

    Let then $B \coloneqq B_+ \cup B_-$, so that $|B| \leq \ell$. Choose $\tup{b}$ as any $\ell$-tuple of vertices of $\bigcup \Xc$ containing $B$ entirely.
    (Note that $\bigcup \Xc$ cannot be empty: otherwise, as $(X_+, X_-, \Xc)$ is splendid, we have $X_+ = X \neq \emptyset$, but this contradicts \Cref{cl:indset-union-is-covering}.)

    \begin{claim}
        \label{cl:span-verification}
        $(X_+, X_-, \Xc) = \Seed(\varphi_+, \varphi_-, \varphi_\sim, \tup{a}\tup{b})$.
    \end{claim}
    \begin{claimproof}[Proof of the claim]
        Since $\tup{b}$ is an $X_+$-cover and $(X_+, X_-, \Xc)$ is splendid, we have $X_+ = \{v \in V(G) \,\mid\, G\models \widehat{\varphi}_+(v, \tup{a}\tup{b})\}$.
        By the same argument, $X_- = \{v \in V(G) \,\mid\, G\models \widehat{\varphi}_-(v, \tup{a}\tup{b})\}$, and so we have $\bigcup\Xc = \{v \in V(G) \,\mid\, G\models \widehat{\varphi}_\Xc(v, \tup{a}\tup{b})\}$.
        Therefore, assertion $\IsConsistent(\tup{a}\tup{b})$ holds: we have $b_i \in \bigcup \Xc$ for each $i \in [\ell]$ by our choice of $\tup{b}$, and \ref{item:suffix-to-param-xc-antichain} precludes the existence of vertices $u_1, u_2 \in \bigcup \Xc$ with $u_1 <_{\tup{a}} u_2$.
        Hence $X_+ = \{v \in V(G) \,\colon\, \varphi_+(v, \tup{a}\tup{b})\}$ and $X_- = \{v \in V(G) \,\colon\, \varphi_-(v, \tup{a}\tup{b})\}$.
        Again, by \ref{item:suffix-to-param-xc-antichain}, $\Xc$ is partitioned into strongly connected components of $H$, so it is defined by the formula $\varphi_\sim$.
    \end{claimproof}

    Since $X \in \Span(X_+, X_-, \Xc)$, \Cref{cl:span-verification} settles the proof of \Cref{lem:suffix-to-param}.
\end{proof}

As mentioned, \Cref{lem:suffixes-to-spans} follows directly from \Cref{lem:def-to-seed,lem:param-to-suffix,lem:suffix-to-param}.

%% file: low-rank-and-flip-reach.tex
\subsection{Proof of the equivalence}\label{sec:flip-reach-proof}

With the combinatorial characterization of sets of low rank we can prove that flip-reachability logic has the low rank definability property on the class of all graphs (\Cref{def:low-rank-definability}).

\begin{lemma}
    \label{lem:low-rank-definability-freach}
    Flip-reachability logic has low rank definability property on the class of all graphs.
\end{lemma}
\begin{proof}
    We proceed similarly to the proof of~\Cref{lem:lrdp-separator} and~\Cref{lem:property-for-fconn}.
    Let $r\in\mathbb{N}$ be a bound on the rank and $\varphi(X, \wtup Y, \tup z)$ be a formula of flip-reachability logic, as in the definition of low rank definability property. As in the proofs of~\Cref{lem:lrdp-separator} and~\Cref{lem:property-for-fconn}, we may assume that $\wtup Y=\emptyset$ and $\tup z=\emptyset$ (and therefore $\varphi$ can be assumed to have only one free set variable $X$).

    In the remainder of the proof we show that there is a formula $\varphi'(x,\tup t)$ of flip-reachability logic, where $x$ and $\tup t$ are free vertex variables, such that the following holds.
    For every graph $G$ and set $A \subseteq V(G)$ such that $G\models \varphi(A)$, there exists a set $A'\subseteq V(G)$ and an evaluation $\tup d$ of variables $\tup t$ such that
    \[G \models \varphi(A')\qquad \textrm{and}\qquad A' = \setof{v \in V(G)}{G \models \varphi'(v, \tup d)}.\]

    The core of our construction is~\Cref{thm:low-rank-structure}, from which we get the formulas $\varphi_+, \varphi_-,\varphi_\sim$ of flip-reachability logic such that the following holds. For every graph $G$ and set $A\subseteq V(G)$ with $\rk(A)\leq r$,
    there is a tuple $\tup a$ such that $A\in \Span(\Seed(\varphi_+, \varphi_-, \varphi_{\sim}, \tup{a}))$. Note that these formulas depend only on $r$.

    We will also use the \emph{types} of flip-reachability logic.
    Let $\Sigma$ be a finite set of unary predicates including the unary predicates contained in $\varphi$.
    Given an $\Sigma$-colored graph $H$, its \emph{$q$-flip-reachability type} is the set of all flip-reachability sentences of quantifier rank at most $q$ that hold in $H$.
    Similarly as for the flip-connectivity logic, there are only finitely many such types and they can be defined in flip-reachability logic.

    Let $q$ be the quantifier rank of formula $\varphi$.
    \begin{claim}
        Let $G$ be a graph. Assume that there is a set $A \subseteq V(G)$ and a tuple $\tup a$ such that $G\models \varphi(A)$ and $A\in \Span(X_+, X_-, \Xx_{\sim})$, where $(X_+, X_-, \Xx_{\sim}) \coloneqq \Seed(\varphi_+, \varphi_-, \varphi_{\sim}, \tup{a})$. Also, assume that $G$ is marked with unary predicates added for atomic types over $\tup a$. Then there also exists a set $A'\subseteq V(G)$ such that the following holds:
        \begin{itemize}[nosep]
            \item $G\models \varphi(A')$;
            \item $A'\in \Span(X_+, X_-, \Xx_{\sim})$; and
            \item for each $q$-flip-connectivity type $\tau$, $A'$ contains at most $q$ parts from $\Xx_{\sim}$ with $q$-flip-reachability type $\tau$ or $A'$ contains all but at most $q$ parts from $\Xx_{\sim}$ with $q$-flip-reachability type $\tau$.
        \end{itemize}
    \end{claim}
    \begin{claimproof}
        We consider a set $A' \in \Span(X_+, X_-, \Xx_{\sim})$ as follows:
        \begin{itemize}[nosep]
            \item if $A$ contains at most $q$ parts from $\Xx_{\sim}$ with $q$-flip-reachability type $\tau$, then $A'$ contains the same parts from $\Xx_{\sim}$ with $q$-flip-connectivity type $\tau$, and
            \item if $A$ contains all but at most $q$ parts from $\Xx_{\sim}$ with $r$-flip-connectivity type $\tau$, then $A'$ contains the same parts from $\Xx_{\sim}$ with $q$-flip-connectivity type $\tau$, and
            \item in all the other cases $A'$ contains $q$ arbitrary parts from $\Xx_{\sim}$ with $q$-flip-connectivity type $\tau$.
        \end{itemize}

        Now observe that $G \models \varphi(A)$ if and only if $G \models \varphi(A')$.
        Indeed, this is a simple Ehrenfeucht--Fraïssé argument.
        Namely, if Duplicator has a winning strategy for a $q$-round EF game in the structure $G_A \coloneqq (G, A)$, then clearly she can modify her strategy to win in the structure $G_{A'} \coloneqq (G, A')$ by always playing in the part of $\Xx_{\sim}$ in $G_{A'}$ that has the same $q$-flip-reachability type as the part of $\Xx_\sim$ in $G_A$ that she would play in the original game.
        Whenever at the end of the $q$ round game we get two tuples $\tup e$ of vertices played in $G_A$ and $\tup f$ of vertices played in $G_{A'}$ then since the seed $(X_+, X_-, \Xx_{\sim})$ is $\tup a$ uniform and we evaluated $q$-flip-reachability types in $G^+$ that contains predicates for atomic types on $\tup a$, then $\tup e$ and $\tup f$ have the same atomic types with respect to the edge relation and flip-reachability predicates.
    \end{claimproof}

    It remains to construct the formula $\phi'(x, \tup t)$ of flip-reachability logic with a tuple of parameters $\tup t$ of bounded size that defines such a set $A'$.
    Indeed, it is enough to use the parameters from $\tup a$ and for every $q$-flip-reachability type $\tau$ we need at most $q$ additional free variables -- one for each part of $\Xx_{\sim}$ with $q$-flip-reachability type $\tau$ which is in $A'$ (or is not in $A'$ in the case when $A'$ contains all but at most $q$ parts of $\Xx_{\sim}$ with $q$-flip-reachability type $\tau$).
    This way, we can define for every $q$-flip-reachability type $\tau$, we can consider a formula $\xi_\tau$ and there is a bounded number of such formulas.
    Therefore, we can write a single formula $\phi'(x, \tup t)$ by extending the tuple $\tup t$ with a finite number of dummy variables; the formula $\psi$ chooses the appropriate formula $\xi$ by testing the equality type of these dummy variables.   
\end{proof}

\Cref{thm:main-freach} follows from \Cref{lem:low-rank-definability-freach} and \Cref{thm:low-rank-quantifier-elimination}.

%% file: conclusions.tex
\section{Conclusions}
\label{sec:conclusions}

We conclude by outlining several possible directions for future work.

\paragraph*{Beyond graphs.} In this work we considered low rank \mso on undirected graphs. However, the definition can be naturally extended to any kind of structures in which a meaningful notion of the rank of a set can be considered. Take for example binary relational structures. For a partition $(X,\wh X)$ of the universe of a structure $\mathbb A$, we can define the bipartite adjacency matrix $\Adj_{\mathbb A}[X,\wh X]$ by placing at the intersection of the row of $u\in X$ and the column of $v\in X$, the atomic type of the pair $(u,v)$. This way, the adjacency matrix is no longer binary, but as discussed in \cref{sec:logics}, we can still measure its diversity: the number of distinct rows plus the number of distinct columns. Now, in low rank \mso over binary structures we would stipulate the set quantification to range only over sets that induce cuts of bounded diversity. For structures of higher arity, say bounded by $k$, one natural definition would be to index rows of $\Adj_{\mathbb A}[X,\wh X]$ by $(<k)$-tuples of elements of $X$, the columns by $(<k)$-tuples of elements of $Y$, and at the intersection of row $\tup u\in X^{<k}$ and column $\tup v\in \wh X^{<k}$ put the atomic type of $\tup u \tup v$.

To see an example of this definition in practice, consider the setting of finite words over a finite alphabet~$\Sigma$. There are two ways of encoding such words as relational structures: either by equipping the set of positions by the total order, or only by the successor relation. Deploying first-order logic \fo on these two encodings yields two logics with different expressive power. It is customary to consider the ordered one as the right notion of \fo on words, mainly due to classic connections with star-free expressions~\cite{McNaughtonPapert71} and aperiodic monoids~\cite{Schutzenberger65}. In contrast, it is not difficult to see that in both encodings, sets of bounded rank can be characterized as unions of a bounded number of intervals of positions. Consequently, regardless of the encoding, low rank \mso on finite words is equivalent to \fo with access to the total order on positions.

Another case is that of directed graphs, which can be seen as binary structures consisting of the vertex set equipped with a (not necessarily symmetric) arc relation. Recall that we proved that in undirected graphs, properties expressible in low rank \mso can be decided in polynomial time (\cref{cor:xp}), and this was a consequence of the equivalence of low rank \mso and flip-reachability logic (\cref{thm:main-freach}). We do not know whether this equivalence holds in general directed graphs, hence the following question is open.

\begin{question}
 Can every property of directed graphs definable in low rank \mso be decided in polynomial~time?
\end{question}

Finally, another interesting class of structures on which low rank \mso can be deployed are matroids. This setting is slightly tricky, as it is natural to take the connectivity function as the concept of the rank of a set, instead of the standard matroid rank function (i.e., maximum size of an independent subset). That is, if $\cal M$ is a matroid with ground set $E$ and $(X,\wh X)$ is a partition of $E$, then the connectivity function of $X$ is $\rkk_{\cal M}(X)+\rkk_{\cal M}(\wh X)-\rkk_{\cal M}(E)$, where $\rkk_{\cal M}(\cdot)$ is the standard rank function of $\cal M$; this way the connectivity function of $X$ and of $\wh X$ are equal. Given the tight links between the structural theory of matroids and the theory of vertex-minors, where sets of small cutrank play a key role, one might suspect that low rank \mso has interesting properties on matroids as well.

\paragraph*{Monadic dependence.} The recent developments on computational aspects of \fo have highlighted the notion of {\em{monadically dependent}} graph classes as a key structural dividing line. In a nutshell, a class of graphs $\Cc$ is monadically dependent if one cannot interpret the class of all graphs in vertex-colored graphs from $\Cc$ using a fixed \fo interpretation; see~\cite[Section 4.1]{Pilipczuk25} for a formal definition and a discussion. The notion of monadic dependence can be naturally defined also for logics other than \fo. What would be then monadic dependence with respect to low rank \mso? It is not hard to prove using the results of Pilipczuk et al.~\cite{separatorModelChecking} that on subgraph-closed classes, monadic dependence with respect to separator logic coincides with the property of excluding a fixed graph as a topological minor. So by analogy, monadic dependence with respect to low rank \mso should be a dense, logically-motivated analogue of excluding topological minors. It must be, however, different from the property of excluding a fixed vertex-minor, for the latter is known not to imply monadic dependence with respect to \fo~\cite{HlinenyP22}.

We pose the following question as an excuse to explore combinatorial properties of graph classes that are monadically dependent with respect to low rank \mso. It is a low rank \mso counterpart of the analogous question posed for \fo (see e.g.~\cite[Conjecture~2]{Pilipczuk25}), which is currently under intensive investigation.

\begin{question}
 Is the model-checking problem for low rank \mso fixed-parameter tractable on every class of graphs that is monadically dependent with respect to low rank \mso?
\end{question}

\paragraph*{Space complexity.} Note that every graph property definable in \fo can be decided in $\mathsf{L}$ (deterministic logarithmic space): a brute-force model-checking algorithm needs only to remember an evaluation of constantly many variables present in the verified sentence. Due to Reingold's result that undirected reachability is in $\mathsf{L}$~\cite{Reingold08}, the same can be said about separator logic and about flip-connectivity logic; hence also about low rank \mso on any graph class of bounded VC dimension, by \cref{thm:main-fconn-positive}. However, the argument breaks for flip-reachability logic, because testing flip-reachability predicates requires solving directed reachability, which is $\mathsf{NL}$-complete. Consequently, thanks to Immerman--Szelepcs\'enyi
 Theorem~\cite{Immerman88,Szelepcsenyi88} and \cref{thm:main-freach}, we can place the model-checking problem for low rank \mso in slicewise $\mathsf{NL}$, but membership in slicewise $\mathsf{L}$ is unclear. We ask the following.

 \begin{question}
  Can every property of undirected graphs definable in low rank \mso be decided in deterministic logarithmic space?
 \end{question}